%% file: main.tex
\documentclass[pdflatex,sn-mathphys-num]{sn-jnl}

\usepackage{iftex}%
\ifPDFTeX
\usepackage[T1]{fontenc}%
\usepackage{lmodern}%
\fi

\usepackage{graphicx}%
\usepackage{amsmath,amssymb,amsfonts}%
\usepackage{amsthm}%
\usepackage{mathtools}%
\usepackage{bm}%
\usepackage{array}%
\usepackage{xurl}%

\theoremstyle{thmstyleone}%
\newtheorem{theorem}{Theorem}[section]%
\newtheorem{lemma}[theorem]{Lemma}%

\theoremstyle{thmstyletwo}%
\newtheorem{remark}{Remark}%

\theoremstyle{thmstylethree}%

\newcommand{\proofstep}[1]{%
  \par\medskip\noindent\textbf{\textit{#1}}\par\nobreak\smallskip\noindent}
\newcommand{\backmatterhead}[1]{%
  \par\medskip\noindent\textbf{#1.}\enspace\ignorespaces}

\begin{document}

\title[Google Play User-Review Sentiment Index]{Statistical Foundations for a Google Play User-Review Sentiment Index: Signal Fusion, Shrinkage, Distributional Validation, and Dynamic Smoothing}

\author*[1]{\fnm{Marco} \sur{Mandap}, \sfx{Ph.D.}}\email{marco.mandap@bulsu.edu.ph}

\affil*[1]{\orgname{Bulacan State University}, \orgaddress{\street{McArthur Highway, Capitol Compound, Guinhawa}, \city{Malolos City}, \postcode{3000}, \state{Bulacan}, \country{Philippines}}}

\abstract{The decision problem is whether valence among the users who write
reviews of an app has changed enough to warrant investigation of a release,
outage, or support issue. We propose a statistical
specification for that decision, applied to written Google Play reviews
collected under a declared locale, language, sort order, retrieval cap, and
window rule. That endpoint returns a selected collection rather than a random
sample, so the index describes review writers under the declared protocol and
not all app users. What is new is not a sentiment classifier but an auditable
interface: an ordinal threshold model is the preferred representation of star
ratings, a calibrated two-signal generalized least-squares estimator is retained
as an interpretable linear approximation when held-out labels support its
conditional-unbiasedness assumption, the estimand is declared before the
estimator, and every downstream module---uncertainty, diagnostics, and
regularization---is tied to that declaration. The framework separates the
latent average of a fixed collected set from the mean of a protocol-defined
population of review writers, so measurement and between-review variation are
not conflated, and treats helpfulness and recency weights as a different
policy target. Shrinkage, cross-app empirical Bayes, histogram and dependence
diagnostics, and dynamic smoothing are specified as separate
modules with defaults and mandatory sensitivity analyses. A
reference implementation, a declared 6000-review three-app collection, and a
synthetic stress test exercise the displayed
equations: fusion reduced error in clean and small-window scenarios, coordinated
contamination caused severe bias and interval failure, and the collected sample
reproduces the abstention rules. No real-review
validation is claimed.}



\maketitle

\section{Introduction}\label{sec:intro}

Online app reviews provide several imperfect signals of user experience: a
one-to-five-star rating, free-form text, helpfulness counts, timestamps, app
versions, and developer replies. A useful index must combine these signals
without treating popularity as certainty, must remain stable in low-volume
windows, and must distinguish temporal change from sampling artifacts. Sentiment
analysis supplies review-level text scores \cite{pang2008opinion,liu2012sentiment,
hutto2014vader}; the remaining problem is statistical aggregation.

The decision problem is monitoring whether valence among observed review writers
changes enough to warrant investigation of a release, outage, or support issue.
The target is protocol-specific: it is not the sentiment of all app users or all
users who submitted star ratings. The paper separates a descriptive latent
average for the reviews in hand from a superpopulation mean for reviews that
would be written under the same protocol. When helpfulness or recency weights are
introduced, the target changes again to a policy-weighted descriptive quantity;
those weights are reported as design choices and compared with an unweighted
baseline.

This manuscript develops that aggregation framework and presents five
self-contained supporting proofs: (i) inverse-variance weighting for combining
star and text signals (Section~\ref{sec:blue}); (ii) Gaussian-conjugate
shrinkage for low-information windows (Section~\ref{sec:shrinkage}); (iii) the
finite-category multinomial limit needed for star distributions
(Section~\ref{sec:empirics}); (iv) the delta method for transformed count
weights (Section~\ref{sec:delta}); and (v) Kalman filtering for a latent
sentiment trend (Section~\ref{sec:kalman}). Each proof makes its assumptions
explicit. The final section maps the results to fields exposed by the
\texttt{google-play-scraper} Python package, corrects common misuses of
helpfulness counts and the Kolmogorov--Smirnov statistic, gives a toy
calculation, and reports an application to a collected sample of written
reviews. Throughout, \emph{framework} denotes the specification as a whole,
\emph{index} the computed window quantity and the published series, and
\emph{toy calculation} the three-review illustration in
Section~\ref{sec:toy-calculation}. The framework is methodological: it does not
claim empirical calibration or predictive validation from an illustrative
calculation, a synthetic study, or a descriptive collection.
Standard mathematical background is provided by
\cite{aitken1936,stein1956,james1961,billingsley1968,billingsley1999,
vandervaart1998,shorackwellner1986,dudley2014,kalman1960}.

App-review research has established that review text is heterogeneous in purpose,
that aspect-level sentiment can disagree with an overall rating, and that review
submission is selected in time and user type
\cite{paganomaalej2013,guzmanmaalej2014,maalejnabil2015,martin2017,
hupavlouzhang2017,lihitt2008}; two recent systematic surveys map that literature
and its evaluation practice
\cite{dabrowskiletier2022,nakamura2022}. Ratings also feed back into visibility, sales,
and future review volume, so the population that generates the rating base is
itself moving \cite{dellarocas2003}. A decade of app-store analytics has met the
rating--text mismatch, version attribution, and the mixed-sentiment review with
data; what it does not supply is a reproducible statement of the estimand, its
uncertainty, and the conditions under which two windows may be compared. The
contribution here is not a new sentiment classifier. It is an auditable interface
between calibrated signals,
protocol-defined estimands, uncertainty calculations, distributional checks, and
regularization, together with executable calculations that expose where the
interface fails.

\section{Inverse-variance weighting is BLUE, and coincides with the Gaussian MLE}\label{sec:blue}

\proofstep{Setup.}
Let $\theta \in \mathbb{R}$ be a fixed unknown parameter, and let
\begin{equation}
X_1 = \theta + \varepsilon_1, \qquad X_2 = \theta + \varepsilon_2
\end{equation}
be two unbiased estimators satisfying $\mathbb{E}[\varepsilon_i] = 0$,
$\operatorname{Var}(\varepsilon_i) = \sigma_i^2$, and
$\varepsilon_1 \perp \varepsilon_2$. Consider the class of linear combinations:
\begin{equation}
\hat{\theta}(w) = wX_1 + (1-w)X_2, \qquad w \in \mathbb{R}.
\end{equation}

\proofstep{Step 1: Unbiasedness is automatic.}
\begin{equation}
\mathbb{E}[\hat{\theta}(w)] = w\theta + (1-w)\theta = \theta, \quad \forall w \in \mathbb{R}.
\end{equation}
Thus, the entire parameterized family is unbiased; minimizing variance is the only
remaining objective.

\proofstep{Step 2: Minimize variance.}
By independence of $\varepsilon_1$ and $\varepsilon_2$,
\begin{equation}
V(w) \coloneqq \operatorname{Var}(\hat{\theta}(w)) = w^2\sigma_1^2 + (1-w)^2\sigma_2^2.
\end{equation}
This is a quadratic function in $w$ with positive leading coefficient
$\sigma_1^2 + \sigma_2^2 > 0$ (strictly convex). Hence, it admits a unique global
minimum where $V'(w) = 0$:
\begin{equation}
\begin{aligned}
V'(w) &= 2w\sigma_1^2 - 2(1-w)\sigma_2^2 = 0,\\
w^*(\sigma_1^2 + \sigma_2^2) &= \sigma_2^2,\\
w^* &= \frac{\sigma_2^2}{\sigma_1^2 + \sigma_2^2}
= \frac{\sigma_1^{-2}}{\sigma_1^{-2} + \sigma_2^{-2}}.
\end{aligned}
\end{equation}

\proofstep{Step 3: Resulting minimum variance.}
Substituting $w^*$ back, with $a = \sigma_1^{-2}$ and $b = \sigma_2^{-2}$:
\begin{equation}
V(w^*) = \left(\frac{a}{a+b}\right)^2 \frac{1}{a}
+ \left(\frac{b}{a+b}\right)^2 \frac{1}{b}
= \frac{a}{(a+b)^2} + \frac{b}{(a+b)^2}
= \frac{1}{a+b} = \frac{1}{\sigma_1^{-2} + \sigma_2^{-2}}.
\end{equation}
Consequently,
\begin{equation}
\hat{\theta}^* = w^* X_1 + (1-w^*)X_2
= \frac{\sigma_1^{-2}X_1 + \sigma_2^{-2}X_2}{\sigma_1^{-2} + \sigma_2^{-2}}
\end{equation}
is the unique variance-minimizing member of the linear unbiased class, i.e., the
best linear unbiased estimator (BLUE). This is the two-source special case of Aitken's generalized least squares theorem.

\proofstep{Equivalence to the Gaussian MLE and UMVUE optimality.}
Suppose also that the positive variances $\sigma_i^2$ are known and
$\varepsilon_i \sim \mathcal{N}(0, \sigma_i^2)$.
The log-likelihood function is
\begin{equation}
\ell(\theta) = -\frac{(x_1 - \theta)^2}{2\sigma_1^2}
- \frac{(x_2 - \theta)^2}{2\sigma_2^2} + \text{const}.
\end{equation}
Setting $\ell'(\theta) = 0$:
\begin{equation}
\frac{x_1 - \theta}{\sigma_1^2} + \frac{x_2 - \theta}{\sigma_2^2} = 0
\implies \hat{\theta}_{\mathrm{MLE}}
= \frac{x_1\sigma_1^{-2} + x_2\sigma_2^{-2}}{\sigma_1^{-2} + \sigma_2^{-2}},
\end{equation}
which is identical to $\hat{\theta}^*$. Evaluating the second derivative gives
$\ell''(\theta) = -(\sigma_1^{-2} + \sigma_2^{-2}) < 0$, confirming a global maximum.

Fisher information is additive over independent observations:
\begin{equation}
I(\theta) = \sigma_1^{-2} + \sigma_2^{-2}.
\end{equation}
Since $\operatorname{Var}(\hat{\theta}^*) = 1/I(\theta)$, the Cram\'er--Rao lower
bound is achieved with exact equality. Therefore, $\hat{\theta}^*$ is BLUE and
uniformly minimum-variance unbiased (UMVUE) among \emph{all} unbiased estimators
(linear or non-linear) in the Gaussian setting. These are known-variance
statements: additivity of the information and attainment of the bound both rely
on $\sigma_i^2$ being fixed and known as assumed above, and neither claim
survives unchanged when the variances are estimated from the same data.

\section{The shrinkage index is the Bayes posterior mean}\label{sec:shrinkage}

\proofstep{Setup.}
Let $I \mid \mu \sim \mathcal{N}(\mu, 1/V)$ denote the raw weighted index, treated
as unbiased for the latent trend $\mu$ with precision $V$. Consider the prior
$\mu \sim \mathcal{N}(C, 1/m)$.

\proofstep{Step 1: Posterior derivation via completing the square.}
The posterior density satisfies:
\begin{equation}
p(\mu \mid I) \propto \exp\!\left[-\tfrac{1}{2} V(I - \mu)^2\right]
\exp\!\left[-\tfrac{1}{2} m(\mu - C)^2\right].
\end{equation}
Expanding the exponent:
\begin{equation}
-\tfrac{1}{2} \Big[(V+m)\mu^2 - 2(VI + mC)\mu + (VI^2 + mC^2)\Big].
\end{equation}
Completing the square in $\mu$:
\begin{equation}
(V+m)\mu^2 - 2(VI + mC)\mu
= (V+m)\left[\mu - \frac{VI + mC}{V+m}\right]^2 - \frac{(VI + mC)^2}{V+m}.
\end{equation}
Because the remainder does not depend on $\mu$, it is absorbed into the
normalization constant. Hence:
\begin{equation}
\mu \mid I \sim \mathcal{N}\!\left(\frac{VI + mC}{V+m}, \, \frac{1}{V+m}\right)
= \mathcal{N}\big(I_{\mathrm{shrunk}}, \, 1/(V+m)\big).
\end{equation}
Precisions add directly ($V+m$), reflecting Gaussian-conjugate information pooling.

\proofstep{Step 2: Posterior mean as Bayes estimator under quadratic loss.}
For any posterior $p(\mu \mid I)$ and any candidate constant
$c \in \mathbb{R}$:
\begin{equation}
\mathbb{E}[(\mu - c)^2 \mid I]
= \mathbb{E}\big[(\mu - \mathbb{E}[\mu \mid I])^2 \mid I\big]
+ (\mathbb{E}[\mu \mid I] - c)^2,
\end{equation}
since the cross-product term
\begin{equation}
2\,\mathbb{E}\big[(\mu - \mathbb{E}[\mu \mid I])
(\mathbb{E}[\mu \mid I] - c) \mid I\big]
= 2(\mathbb{E}[\mu \mid I] - c)\,\mathbb{E}[\mu - \mathbb{E}[\mu \mid I] \mid I] = 0.
\end{equation}
The squared deviation is uniquely minimized at $c = \mathbb{E}[\mu \mid I]$.
Hence $I_{\mathrm{shrunk}} = \mathbb{E}[\mu \mid I]$ is the optimal Bayes estimator.

\proofstep{Step 3: Exact frequentist risk-dominance condition.}
Let $a \coloneqq V/(V+m)$, so that $I_{\mathrm{shrunk}} = aI + (1-a)C$.
Evaluating under the true fixed parameter $\mu$:
\begin{equation}
\mathrm{Bias} = \mathbb{E}[I_{\mathrm{shrunk}}] - \mu
= a\mu + (1-a)C - \mu = -(1-a)(\mu - C),
\end{equation}
\begin{equation}
\operatorname{Var}(I_{\mathrm{shrunk}}) = \frac{a^2}{V}.
\end{equation}
The mean squared errors are:
\begin{equation}
\mathrm{MSE}(I_{\mathrm{shrunk}}) = (1-a)^2(\mu - C)^2 + \frac{a^2}{V},
\qquad \mathrm{MSE}(I) = \frac{1}{V}.
\end{equation}
Demanding $\mathrm{MSE}(I_{\mathrm{shrunk}}) < \mathrm{MSE}(I)$:
\begin{equation}
(1-a)^2(\mu - C)^2 < \frac{1 - a^2}{V} = \frac{(1-a)(1+a)}{V}
\implies (\mu - C)^2 < \frac{1+a}{V(1-a)}.
\end{equation}
Substituting $a = V/(V+m)$, where $1-a = m/(V+m)$ and $1+a = (2V+m)/(V+m)$:
\begin{equation}
\boxed{(\mu - C)^2 < \frac{2V + m}{Vm}}
\end{equation}
The boxed inequality is the exact boundary for MSE dominance. James--Stein
shrinkage attains uniform risk dominance over the MLE for $p \ge 3$; in one
dimension, no shrinkage estimator dominates uniformly over $\mu \in \mathbb{R}$.
Risk degrades relative to the raw index whenever $\mu$ lies outside this critical
radius, so the 1-D shrinkage estimator is Bayes-optimal under the specified prior
but not uniformly dominant in a frequentist sense.

This boundary is a theoretical risk comparison, not an observable decision rule,
because the true $\mu$ is unknown. The algebra of Step~3 requires only that $I$
be unbiased with finite, correctly specified variance $1/V$, and that $(V,C,m)$
be fixed; Gaussianity is required for the conjugate posterior interpretation in
Steps~1--2, not for this bias--variance calculation. Estimated precision,
calibration bias, selection, dependence, temporal drift, and data-dependent
hyperparameters turn the displayed boundary into a sensitivity diagnostic. Since
its squared radius is $1/V+2/m$, a precise window far from an overly concentrated
prior---for example, after an outage, controversial release, coordinated review
campaign, or major improvement---is the practical setting in which shrinkage
bias is most likely to exceed its variance reduction.

\section{Finite-category distributional limit}\label{sec:empirics}

The application requires a discrete five-category limit rather than the
continuous-null Kolmogorov distribution; Remark~\ref{rem:why-discrete} records
why the probability-integral-transform route is unavailable here. Let
$\bm C_n=(C_{n1},\ldots,C_{nK})^\top\sim
\operatorname{Multinomial}(n,\bm p)$, where
$\bm p=(p_1,\ldots,p_K)^\top$, $p_k\ge0$, and
$\sum_k p_k=1$. Define $\widehat{\bm p}_n=\bm C_n/n$.

\begin{theorem}[Multinomial proportion limit]\label{thm:multinomial-limit}
As $n\to\infty$,
\begin{equation}
\sqrt n(\widehat{\bm p}_n-\bm p)
\rightsquigarrow
\mathcal N_K\!\left(\bm0,
\operatorname{diag}(\bm p)-\bm p\bm p^\top\right).
\label{eq:multinomial-limit-main}
\end{equation}
If $\widehat F_{nj}=\sum_{k=1}^j\widehat p_{nk}$ and
$F_j=\sum_{k=1}^j p_k$, then the limiting covariance of cumulative coordinates
$j,l\in\{1,\ldots,K-1\}$ is
\begin{equation}
\operatorname{Cov}(Z_j,Z_l)=F_{\min(j,l)}-F_jF_l.
\label{eq:cumulative-multinomial-covariance}
\end{equation}
\end{theorem}

\begin{proof}
Write the category of observation $i$ as the one-hot vector
$\bm Y_i\in\{\bm e_1,\ldots,\bm e_K\}$. Then
$\widehat{\bm p}_n=n^{-1}\sum_i\bm Y_i$,
$\mathbb E(\bm Y_i)=\bm p$, and
$\operatorname{Var}(\bm Y_i)=
\operatorname{diag}(\bm p)-\bm p\bm p^\top$. The multivariate central limit
theorem gives equation~\eqref{eq:multinomial-limit-main}. Cumulative proportions
are the linear transformation $\bm A\widehat{\bm p}_n$, where
$A_{jk}=\mathbf1\{k\le j\}$. Therefore their limiting covariance is
$\bm A[\operatorname{diag}(\bm p)-\bm p\bm p^\top]\bm A^\top$, whose
$(j,l)$ entry is equation~\eqref{eq:cumulative-multinomial-covariance}.
\end{proof}

For fixed $K$, the maximum cumulative deviation is a continuous function of this
finite Gaussian vector. Its null distribution depends on $\bm p$ and is therefore
calibrated by a multinomial parametric bootstrap or an exact/permutation method
for discrete outcomes, not by continuous-distribution Kolmogorov--Smirnov
critical values. Section~\ref{sec:sentiment-index} applies this result to written
versus rating-only star counts.

\begin{remark}\label{rem:why-discrete}
The continuous-null theory that this result replaces rests on the probability
integral transform: when $F$ is continuous, $F(X)$ is uniform and the quantile
representation $\hat F_n=G_n\circ F$ reduces the empirical process to the uniform
case, which is what yields distribution-free critical values
\cite{shorackwellner1986,billingsley1968}. A five-category rating distribution
has four atoms, so the transform is stochastically larger than uniform and no
such reduction exists; the finite-category limit above, calibrated by the exact
or bootstrap methods of \cite{conover1972,arnoldemerson2011}, is the theory the
application actually needs. Continuous critical values are retained only as the
misuse this framework corrects.
\end{remark}

\section{The delta method and analysis of \texorpdfstring{$\log(1+U)$}{log(1+U)}}\label{sec:delta}

\begin{theorem}[Univariate delta method]\label{thm:delta}
Let $\{X_n\}$ be a sequence of random variables such that
$\sqrt{n}(X_n - \theta) \Rightarrow \mathcal{N}(0, \sigma^2)$. If $g$ is
differentiable at $\theta$, then:
\begin{equation}
\sqrt{n}\big(g(X_n) - g(\theta)\big)
\Rightarrow \mathcal{N}\big(0, [g'(\theta)]^2 \sigma^2\big).
\end{equation}
\end{theorem}

\begin{proof}
Taylor expanding $g$ around $\theta$:
\begin{equation}
g(x) = g(\theta) + g'(\theta)(x - \theta) + R(x),
\quad \text{where } \lim_{x \to \theta} \frac{R(x)}{x - \theta} = 0.
\end{equation}
Substituting $X_n$:
\begin{equation}
\sqrt{n}\big(g(X_n) - g(\theta)\big)
= g'(\theta)\sqrt{n}(X_n - \theta)
+ \underbrace{\sqrt{n}(X_n - \theta)}_{O_p(1)}
\cdot \underbrace{\frac{R(X_n)}{X_n - \theta}}_{o_p(1)}.
\end{equation}
Since $X_n - \theta = O_p(n^{-1/2}) \xrightarrow{p} 0$, differentiability ensures
$R(X_n)/(X_n - \theta) \xrightarrow{p} 0$. By Slutsky's theorem, the remainder is
$o_p(1)$, concluding:
\begin{equation}
\sqrt{n}\big(g(X_n) - g(\theta)\big)
= g'(\theta)\sqrt{n}(X_n - \theta) + o_p(1)
\Rightarrow \mathcal{N}\big(0, [g'(\theta)]^2 \sigma^2\big).
\end{equation}
When $g'(\theta)=0$, the first-order limit is degenerate; a higher-order delta
method may be needed for a non-degenerate approximation.
\end{proof}

\proofstep{Application to Poisson counts.}
For $U \sim \operatorname{Poisson}(\lambda)$ where $\operatorname{Var}(U) = \lambda$,
the first-order Taylor propagation of error yields:
\begin{equation}
\operatorname{Var}(g(U)) \approx [g'(\lambda)]^2 \operatorname{Var}(U)
= [g'(\lambda)]^2 \lambda.
\end{equation}
Setting this equal to a positive constant $c$ and choosing the increasing
solution:
\begin{equation}
g'(\lambda) = \sqrt{c}\,\lambda^{-1/2} \implies g(\lambda) = 2\sqrt{c\lambda} + C_0.
\end{equation}
Thus, the increasing variance-stabilizing family is
$g(\lambda)=a\sqrt{\lambda}+b$ with $a>0$, and this family is unique up to a
positive scaling and an additive translation: no other increasing transform
stabilizes the Poisson variance to first order.

\proofstep{Asymptotic variance of $\log(1+U)$.}
For $g(\lambda) = \log(1 + \lambda)$, the first derivative is
$g'(\lambda) = \frac{1}{1 + \lambda}$, yielding asymptotic variance:
\begin{equation}
[g'(\lambda)]^2 \lambda = \frac{\lambda}{(1 + \lambda)^2}
\xrightarrow{\lambda \to \infty} 0.
\end{equation}
This variance decays to zero as $\lambda$ grows, so $\log(1+U)$ is an
\emph{over-stabilizing concave dampener}.

Under the common unit-variance normalization, the Anscombe-corrected transform
is $2\sqrt{U + 3/8}$, with the correction inside the root \cite{anscombe1948}.
The same paper derives the corresponding transform for negative-binomial data,
which is the relevant reference for an overdispersed count such as helpfulness
votes; because that transform requires an estimated dispersion parameter, the
framework keeps helpfulness weighting as a declared policy choice rather than a
stabilized variance quantity. The logarithmic transform
$\log(1+U)$ damps large-count
outliers because its derivative decays as $\mathcal{O}(\lambda^{-1})$ rather than
$\mathcal{O}(\lambda^{-1/2})$. This Poisson calculation motivates a concave
dampener but does not establish that review helpfulness counts are Poisson. Nor
does any of the calculus above apply to the prominence weight actually used in
the framework (Section~\ref{sec:sentiment-index}): that weight is a capped affine function of $\log(1+U)$, its
derivative vanishes above the cap, and the cap is a bounded-influence policy
constraint, not a variance-stabilizing transformation.

\section{The Kalman filter for the local-level model}\label{sec:kalman}

Consider the heteroskedastic local-level system:
\begin{equation}
\mu_t = \mu_{t-1} + \eta_t, \qquad \eta_t \stackrel{\mathrm{iid}}{\sim} \mathcal{N}(0, \tau^2),
\end{equation}
\begin{equation}
I_t = \mu_t + \varepsilon_t, \qquad
\varepsilon_t \sim \mathcal{N}(0, R_t), \quad R_t>0,
\end{equation}
with prior $\mu_0 \sim \mathcal{N}(m_0, P_0)$. The state innovations and
observation errors are mutually independent across time, and $R_t$ is treated as
known for the exact Gaussian recursion.

\begin{lemma}[Gaussian conditional distributions]\label{lem:gauss-cond}
Let $(X, Y)^\top$ be jointly Gaussian with $\Sigma_{YY}$ invertible:
\begin{equation}
\begin{pmatrix} X \\ Y \end{pmatrix} \sim \mathcal{N}\left(
\begin{pmatrix} \mu_X \\ \mu_Y \end{pmatrix},
\begin{pmatrix} \Sigma_{XX} & \Sigma_{XY} \\ \Sigma_{YX} & \Sigma_{YY} \end{pmatrix}
\right).
\end{equation}
Then the conditional distribution is:
\begin{equation}
X \mid Y \sim \mathcal{N}\big(\mu_X + \Sigma_{XY}\Sigma_{YY}^{-1}(Y - \mu_Y), \,
\Sigma_{XX} - \Sigma_{XY}\Sigma_{YY}^{-1}\Sigma_{YX}\big).
\end{equation}
\end{lemma}

\begin{proof}
Define the auxiliary variable
$Z \coloneqq X - \Sigma_{XY}\Sigma_{YY}^{-1}(Y - \mu_Y) - \mu_X$. Because $Z$ is an
affine combination of jointly Gaussian variables, $(Z, Y)^\top$ is jointly Gaussian.
The cross-covariance is:
\begin{equation}
\operatorname{Cov}(Z, Y)
= \operatorname{Cov}(X, Y) - \Sigma_{XY}\Sigma_{YY}^{-1}\operatorname{Var}(Y)
= \Sigma_{XY} - \Sigma_{XY}\Sigma_{YY}^{-1}\Sigma_{YY} = 0.
\end{equation}
For jointly Gaussian vectors, uncorrelatedness implies independence: $Z \perp Y$.
Therefore:
\begin{equation}
\mathbb{E}[Z \mid Y] = \mathbb{E}[Z] = 0,
\end{equation}
\begin{equation}
\operatorname{Var}(Z \mid Y) = \operatorname{Var}(Z)
= \Sigma_{XX} - \Sigma_{XY}\Sigma_{YY}^{-1}\Sigma_{YX}.
\end{equation}
Re-substituting $X = Z + \mu_X + \Sigma_{XY}\Sigma_{YY}^{-1}(Y - \mu_Y)$ gives the
stated conditional moments.
\end{proof}

\proofstep{Exact MMSE optimality via induction.}
Any linear transformation of Gaussian primitives remains Gaussian, so the joint
vector $(\mu_t,I_{1:t})$ is Gaussian and every conditional expectation
$\mathbb{E}[\mu_t \mid I_{1:t}]$ is linear in the observations. Because the
conditional expectation $\mathbb{E}[X \mid Y]$ uniquely minimizes the mean squared
error over \emph{all} measurable functions (linear or non-linear), the linear
Kalman recursion achieves exact Bayes MMSE optimality.

\proofstep{Inductive recursion.}
\begin{itemize}
\item \textbf{Base case ($t=0$):}
\begin{equation}
\hat{\mu}_{0\mid 0} = m_0, \qquad P_{0\mid 0} = P_0.
\end{equation}
\item \textbf{Time update (predict):}
Assume $\mu_{t-1} \mid I_{1:t-1} \sim \mathcal{N}(\hat{\mu}_{t-1\mid t-1}, P_{t-1\mid t-1})$.
Since $\eta_t \perp I_{1:t-1}$:
\begin{equation}
\hat{\mu}_{t\mid t-1} = \mathbb{E}[\mu_{t-1} + \eta_t \mid I_{1:t-1}]
= \hat{\mu}_{t-1\mid t-1},
\end{equation}
\begin{equation}
P_{t\mid t-1} = \operatorname{Var}(\mu_{t-1} + \eta_t \mid I_{1:t-1})
= P_{t-1\mid t-1} + \tau^2.
\end{equation}
\item \textbf{Measurement update (correct):}
Conditional on $I_{1:t-1}$, the vector $(\mu_t, I_t)^\top$ is jointly Gaussian:
\begin{equation}
\operatorname{Cov}(\mu_t, I_t \mid I_{1:t-1})
= \operatorname{Var}(\mu_t \mid I_{1:t-1}) = P_{t\mid t-1},
\end{equation}
\begin{equation}
\operatorname{Var}(I_t \mid I_{1:t-1}) = P_{t\mid t-1} + R_t > 0.
\end{equation}
Applying Lemma~\ref{lem:gauss-cond}:
\begin{equation}
K_t \coloneqq \frac{P_{t\mid t-1}}{P_{t\mid t-1} + R_t},
\end{equation}
\begin{equation}
\hat{\mu}_{t\mid t} = \hat{\mu}_{t\mid t-1} + K_t\big(I_t - \hat{\mu}_{t\mid t-1}\big),
\end{equation}
\begin{equation}
P_{t\mid t} = P_{t\mid t-1} - K_t P_{t\mid t-1}
= \frac{P_{t\mid t-1}R_t}{P_{t\mid t-1} + R_t}.
\end{equation}
\end{itemize}
This completes the induction. The conditional law is
$\mu_t \mid I_{1:t} \sim \mathcal{N}(\hat{\mu}_{t\mid t}, P_{t\mid t})$, and its
mean $\hat{\mu}_{t\mid t}$ is the exact Bayes MMSE estimator.

\input{google-play-sentiment-index}

\section{Conclusion}\label{sec:conclusion}

This paper specifies a review-sentiment framework as a sequence of explicit
statistical choices rather than as a universal app score. An ordinal star model
or a validated linear approximation combines rating and text information; a
finite-collection latent average and a protocol-defined writer-population mean
are reported separately; optional helpfulness and recency weights define a third,
policy-weighted target. Measurement variance, between-review heterogeneity, and
cross-review dependence are likewise separated. Prior-based shrinkage, genuine
cross-app empirical Bayes, and dynamic smoothing are alternative reporting
layers rather than interchangeable definitions of the index.

The accompanying reference calculations and deterministic synthetic harness test
the displayed equations, sensitivity to prior strength, selection, small windows,
heavy tails, coordinated contamination, and temporal shocks. They show that
fusion can reduce error in the declared data-generating processes, but they also
show that a coordinated attack can defeat every point estimator and collapse
nominal coverage. A declared three-app collection of 6000 written reviews runs
the same pipeline on real records and reproduces the abstention side of the
specification: cross-app empirical Bayes is not estimable at three apps, the
writing-propensity guardrail and the discrete rating gap are unavailable from a
single cumulative snapshot, and two of three profile likelihoods for the process
variance pile up at their lower boundary. Accordingly, the framework includes
comparability and abstention flags rather than treating high volume as
certainty. The framework remains a methodological specification, not an
empirically validated Google Play index. Substantive use still requires
independently constructed labels, a documented collection protocol, real-app
validation, and comparison with star-only, text-only, and unweighted baselines.

\backmatter

\backmatterhead{Supplementary information}
The reference module, tests, synthetic harness, collection and real-data
analysis scripts, generated CSV summaries, collected sample with its manifest,
and figure sources accompany the manuscript in the project repository.

\backmatterhead{Acknowledgements}
Not applicable.

\section*{Statements and Declarations}

\begin{itemize}
\item \textbf{Funding:} No funding was received to assist with the preparation of this manuscript.
\item \textbf{Competing interests:} The authors have no competing interests to declare that are relevant to the content of this article.
\item \textbf{Data availability:} A descriptive sample of 6000 public written Google Play reviews (three apps, 2000 each, protocol $\pi=$ en/us/newest/calendar-day/2000, collected 2026-10-04 UTC) is stored in \texttt{data/} together with its collection manifest and app-level totals; it contains public review text and metadata only, and no user display names are retained. Derived summaries are supplied in \texttt{real-data-output/}. The three-review calculation is explicitly synthetic. Deterministic simulation summaries are supplied in \texttt{simulation-output/} and are regenerated from seed 20261003 by the accompanying harness.
\item \textbf{Ethics approval and consent to participate:} Not applicable.
\item \textbf{Consent for publication:} Not applicable.
\item \textbf{Materials availability:} Not applicable.
\item \textbf{Code availability:} The accompanying \texttt{code/sentiment\_index.py} module, unit tests, deterministic synthetic harness, collection script, real-data analysis script, CSV summaries, and figure-generation code reproduce the numerical calculations reported in Section~\ref{sec:sentiment-index}. Collection uses the vendored \texttt{google-play-scraper} release 1.2.7
(commit \texttt{ce1df6d}) under the declared protocol; no validated sentiment
classifier is included.
\item \textbf{Author contributions:} The author developed the index specification, prepared the proofs, and wrote the manuscript.
\end{itemize}

\bibliography{sn-bibliography}

\end{document}

%% file: google-play-sentiment-index.tex
\section{A Google Play user-review sentiment-index framework}
\label{sec:sentiment-index}

This section converts the preceding results into an estimable methodological
framework using fields exposed by the unofficial \texttt{google-play-scraper}
Python package \cite{googlePlayScraper}. The framework separates the primary
inferential estimand from optional policy weights, distinguishes observed fields
from modeling assumptions, and treats all uncertainty calculations as conditional
on a declared data-collection protocol.

\subsection{Observable fields, provenance, and validation}

The app-detail call exposes the app identifier, aggregate score, counts of ratings
and reviews, the five-category star-rating histogram, current version, and update
time. The review call adds the review text, star score, helpfulness count,
timestamp, reviewed app version, and any developer reply.
Table~\ref{tab:review-fields} maps these fields to their statistical roles.

\begin{table*}[htbp]
\caption{Google Play fields used in the sentiment-index framework.}
\label{tab:review-fields}
\centering
\footnotesize
\begin{tabular}{>{\raggedright\arraybackslash}p{0.23\textwidth}
                >{\raggedright\arraybackslash}p{0.25\textwidth}
                >{\raggedright\arraybackslash}p{0.41\textwidth}}
\hline
\textbf{Field} & \textbf{Statistical role} & \textbf{Use and limitation} \\
\hline
\texttt{content} & Text input & A calibrated sentiment model produces a score; uncertainty must be estimated by that model or by held-out residuals. \\
\texttt{score} & Ordinal rating signal & Maps a one-to-five-star response to a working scale; equal spacing is a modeling assumption, and the discreteness that equal spacing leaves untreated is why the ordinal threshold model is preferred. \\
\texttt{thumbsUpCount} & Social endorsement & May define an optional prominence weight, but lacks impressions and negative-vote denominators, so a helpfulness ratio is not a Wilson proportion \cite{wilson1927}. \\
\texttt{at} & Review time & Defines time windows and optional recency weights. \\
\texttt{reviewCreatedVersion}, \texttt{appVersion} & Release context & Supports version-level grouping and associational change-point analyses. \\
\texttt{replyContent}, \texttt{repliedAt} & Developer response & Supplies optional responsiveness covariates after timestamp and missing-value validation. \\
\texttt{histogram} & All-rating distribution & Diagnoses differences between written-review ratings and all submitted ratings; it does not supply text for rating-only users. \\
\hline
\end{tabular}
\end{table*}

The review endpoint paginates results and permits alternative sort orders. A call
returning the newest or most relevant reviews is therefore a selected collection,
not automatically a simple random sample. Retrieving all exposed reviews reduces
truncation but does not correct platform-level selection, and high-volume apps may
still expose only a ranked subset. Each collection must record both the number
retrieved and the contemporaneous app-level totals.

Let the collection protocol be
$\pi=(\ell,c,o,\mathcal{W},L)$, comprising language or locale $\ell$, country $c$,
sort order $o$, window rule $\mathcal{W}$, and retrieval cap $L$. An empirical
analysis should pin the scraper release and commit, record collection times and
query arguments, deduplicate by \texttt{reviewId}, and retain repeated snapshots
so edits can be detected. Developer replies are post-review events and must not be
used as baseline covariates; response-time analyses require \texttt{repliedAt}
after \texttt{at}. The cited documentation contains a counterexample to that
ordering.

\subsection{Estimands and review-level measurement}

Let $\mathcal{W}_t(\pi)$ be the written reviews retained in window $t$. For fixed
nonnegative weights $\omega_i$, the finite-collection descriptive target is

\begin{equation}
\bar\theta_{t,\omega}^{\mathrm{obs}}
=\frac{\sum_{i\in\mathcal{W}_t(\pi)}\omega_i\theta_i}
       {\sum_{i\in\mathcal{W}_t(\pi)}\omega_i}.
\label{eq:finite-review-estimand}
\end{equation}

The primary inferential target used for comparisons is instead the
protocol-defined writer-population mean

\begin{equation}
\mu_t^{\pi}
=\mathbb{E}\!\left(\theta_i\mid R_i=1,\ t_i\in t,\ \pi\right),
\label{eq:review-estimand}
\end{equation}

where $R_i=1$ indicates that a written review is observed. Neither target equals
sentiment among all app users or all rating submitters. Equation
\eqref{eq:finite-review-estimand} conditions on the collected records; equation
\eqref{eq:review-estimand} treats them as observations from the population that
would write under the same protocol. Their uncertainty calculations therefore
differ.

For star rating $r_i\in\{1,2,3,4,5\}$, the preferred measurement model is a
cumulative-link model

\begin{equation}
\mathbb{P}(r_i\le k\mid\theta_i,\ell_i)
=F(\tau_{k\ell_i}-\theta_i),\qquad k=1,\ldots,4,
\label{eq:ordinal-star-model}
\end{equation}

with ordered thresholds and a fixed latent-location and latent-scale convention
\cite{mccullagh1980}. A text model supplies

\begin{equation}
s_i=\lambda_{0\ell_i}+\lambda_{1\ell_i}\theta_i+u_i,
\qquad
\operatorname{Var}(u_i\mid\bm x_i)=\sigma_T^2(\bm x_i),
\label{eq:text-measurement-model}
\end{equation}

where useful variance covariates include review length, lexicon coverage,
language, app domain, and proximity to the score boundary. The last covariate is
not optional in practice: rule-based compound scores such as VADER saturate near
their endpoints and degrade on domain-specific vocabulary, so
$\sigma_T^2(\bm x_i)$ is strongly heteroskedastic in the score itself
\cite{hutto2014vader}. Fine-tuned transformer encoders occupy the same position
in equation~\eqref{eq:text-measurement-model} and inherit exactly the same
requirement \cite{devlin2019bert,reimers2019sbert}: a score from any text model
is delivered already needing calibration, a variance model, and held-out
coverage. Text-only annotations are another text-side measurement,
not an error-free gold standard. Identifiability requires explicit scale
constraints, an additional conditionally independent indicator, or a sensitivity
analysis over residual star--text correlation.

That last requirement is substantive rather than technical. A practically
available label is produced by an annotator reading the text, so it is far more
correlated with the text signal than with the star rating; a covariance model
fitted to residuals around such a construct would inherit one of the two signals'
view of valence and become circular. Validation records must therefore be
labeled with the star rating hidden, by more than one independent rater, and
scored for inter-rater agreement, or $\bm\Sigma_i$ must be reported as a
sensitivity band over the residual star--text correlation rather than as a point
estimate.

Label-set size is part of that protocol rather than an afterthought. Report the
label count in each rating cell of each language: a cell whose count does not
exceed the number of free calibration parameters cannot identify its own
threshold map, so such cells are pooled under a declared rule or flagged
non-estimable instead of being fitted. Report the standard errors of the fitted
calibration and the held-out coverage by cell against a coverage tolerance
declared in $\pi$ before any index is computed; a cell outside that tolerance
receives more labels or a simpler calibration, not a post-hoc adjustment.

The joint ordinal--continuous model may be estimated by marginal likelihood,
integrating over a distribution for $\theta_i$. In a multilingual window each
review carries its own language-specific calibration
$(\lambda_{0\ell_i},\lambda_{1\ell_i})$ and its own thresholds
$\tau_{k\ell_i}$; a single pooled calibration across languages is not permitted,
and locale and language are therefore part of the protocol $\pi$ rather than
reporting metadata. For a transparent linear approximation, define
$\rho_i=(r_i-3)/2$ and obtain held-out calibrated signals
$z_{\rho i}$ and $z_{s i}$. Their required working assumption is stronger than
probability calibration:

\begin{equation}
\bm y_i=(z_{\rho i},z_{s i})^\top=\theta_i\bm 1+\bm e_i,
\quad
\mathbb{E}(\bm e_i\mid\theta_i,\bm x_i)=\bm 0,
\quad
\operatorname{Var}(\bm e_i\mid\theta_i,\bm x_i)=\bm\Sigma_i.
\label{eq:calibrated-measurement}
\end{equation}

The positive-definite $\bm\Sigma_i$ must be fitted on independent validation
records; it cannot be estimated from review $i$ alone. Under that working
assumption the generalized least-squares approximation is

\begin{equation}
\bm a_i=\frac{\bm\Sigma_i^{-1}\bm 1}
 {\bm 1^{\top}\bm\Sigma_i^{-1}\bm 1},
 \qquad
\hat\theta_i=\bm a_i^{\top}\bm y_i,
 \qquad
\nu_i=\bm a_i^{\top}\bm\Sigma_i\bm a_i.
\label{eq:review-fusion}
\end{equation}

Because $\rho_i$ is a coarsened, non-monotone function of $\theta_i$, the
classical assumption in equation~\eqref{eq:calibrated-measurement} is a modeling
choice rather than a consequence of the data: the residual of a thresholded
signal is a Berkson-type error, mechanically negatively correlated with
$\theta_i$. The weights in equation~\eqref{eq:review-fusion} are therefore
working weights, efficient only under that choice, and the ordinal model in
equation~\eqref{eq:ordinal-star-model} is the specification the framework
recommends whenever it is estimable. The ceiling that the discreteness creates
can be read off the fitted thresholds: a window dominated by five-star ratings
has little residual variation left with which to identify them, so an all-five-star
window is reported as uninformative about the upper thresholds rather than as a
precise estimate of high valence.

Writing $\bm\Sigma_i$ with entries
$(\sigma_{\rho i}^2,\sigma_{s i}^2,\sigma_{\rho s,i})$, the two-signal convex
solution is available without a numerical quadratic-programming solver:

\begin{equation}
a_{\rho i}^{+}=\operatorname{clip}_{[0,1]}\!\left(
\frac{\sigma_{s i}^2-\sigma_{\rho s,i}}
{\sigma_{\rho i}^2+\sigma_{s i}^2-2\sigma_{\rho s,i}}\right),
\qquad a_{s i}^{+}=1-a_{\rho i}^{+}.
\label{eq:constrained-fusion}
\end{equation}

The clipping step is what makes equation~\eqref{eq:constrained-fusion} available
without a quadratic-programming solver, and it can set a weight to exactly zero
when the unconstrained optimum is negative; such zero weights are reported as
boundary solutions, not as missing values.

If validation data are unavailable, a prespecified convex combination is more
transparent than a spurious precision estimate. The disagreement
$\delta_i=|z_{\rho i}-z_{s i}|$ flags rating--text mismatch, the failure case that
the opinion-mining literature treats at length
\cite{pang2008opinion,liu2012sentiment}, but it does not identify
which signal is wrong. A separately reported sensitivity analysis may inflate
the uncertainty,

\begin{equation}
g_i=1+\gamma(\delta_i-c_{\delta})_+^2,
\qquad \nu_i^{\mathrm{rob}}=g_i\nu_i,
\qquad (x)_+=\max(x,0),
\label{eq:disagreement-inflation}
\end{equation}

with $(c_{\delta},\gamma)$ chosen by held-out coverage. This device changes the
uncertainty supplied to later stages; it does not alter the raw point estimate or
resolve sarcasm and mixed-aspect semantics. Those require aspect- or
version-aware labels and models.

The deterministic harness evaluates that choice against the
conditionally-unbiased working model of
equation~\eqref{eq:calibrated-measurement}, selecting $(c_{\delta},\gamma)$ on
one half of the replicates by held-out coverage and scoring it on the other. In
the arm where the declared working covariance is correct, evaluation coverage is
0.947 without the device and 0.948 with it, so the device is inert. In the arm
where a declared 30\% of reviews carry text noise with standard deviation 0.55
rather than 0.15, the same working covariance covers only 0.887 of the
replicates and the held-out pair $(c_{\delta},\gamma)=(0.5,12)$ raises coverage
to 0.937 at an interval 11\% wider. The device is therefore a declared
sensitivity device with a demonstrated operating range, not a general coverage
correction; both arms, with their selected parameters, are in
\texttt{disagreement-inflation.csv}.

\subsection{Primary and policy-weighted window indices}

For the primary written-review estimand in equation~\eqref{eq:review-estimand},
set $\omega_i=1$. An optional policy-weighted descriptive index can instead use

\begin{equation}
\omega_i=q_i d_i,
\qquad
q_i=1+\alpha\min\{\log(1+u_i),q_{\max}\},
\qquad
d_i=\exp\{-\lambda(T_t-t_i)\},
\label{eq:review-weight}
\end{equation}

where $u_i$ is the helpfulness count, $T_t$ is the endpoint of window $t$,
$\alpha\ge0$, $q_{\max}<\infty$, and $\lambda\ge0$ has reciprocal units to the
chosen time scale. The leading 1 prevents zero-vote reviews from receiving zero
weight. These weights define a prominence- and recency-weighted target; they are
not a correction for representative sampling. The cap bounds the relative
prominence weight by $1+\alpha q_{\max}$; it is a policy constraint, not a
variance-stabilizing result or a claim that helpfulness counts are Poisson:
Theorem~\ref{thm:delta} and its application to $\log(1+U)$ in
Section~\ref{sec:delta} make that distinction explicit. The
half-life of the recency component is $\log(2)/\lambda$ and must be reported with
the window duration. Because the discount is applied inside an already
time-bounded window, the effective kernel is a truncated exponential whose shape
depends jointly on $\lambda$ and the window width: a half-life comparable to the
window duration concentrates weight on the window edge, so $(\lambda,\mathcal{W})$
are declared together rather than separately.

For either declared target, define

\begin{equation}
I_t=\frac{\sum_{i\in\mathcal{W}_t(\pi)}\omega_i\hat\theta_i}
          {\sum_{i\in\mathcal{W}_t(\pi)}\omega_i}.
\label{eq:raw-sentiment-index}
\end{equation}

For the fixed records in equation~\eqref{eq:finite-review-estimand}, only
measurement error contributes to uncertainty:

\begin{equation}
M_t^2=
\frac{\sum_{i\in\mathcal{W}_t(\pi)}\omega_i^2\nu_i}
     {\left(\sum_{i\in\mathcal{W}_t(\pi)}\omega_i\right)^2}.
\label{eq:finite-index-variance}
\end{equation}

For the writer-population mean in equation~\eqref{eq:review-estimand}, use the
working random-effects model

\begin{equation}
\hat\theta_i=\mu_t^{\pi}+b_i+e_i,
\qquad
\operatorname{Var}(b_i)=\kappa_t^2,
\qquad
\operatorname{Var}(e_i)=\nu_i.
\label{eq:review-random-effects}
\end{equation}

The corresponding linearized variance is

\begin{equation}
S_t^2=
\frac{\sum_{i\in\mathcal{W}_t(\pi)}\omega_i^2
      (\nu_i+\hat\kappa_t^2)}
     {\left(\sum_{i\in\mathcal{W}_t(\pi)}\omega_i\right)^2}.
\label{eq:index-variance}
\end{equation}

Thus $M_t^2$ answers uncertainty about the collected reviews, whereas $S_t^2$
answers uncertainty about the writer-population mean. They must not be
interchanged. A DerSimonian--Laird diagnostic, for example, uses
$a_i=\nu_i^{-1}$,
$\bar\theta_a=\sum_i a_i\hat\theta_i/\sum_i a_i$, and

\begin{equation}
\hat\kappa_{t,\mathrm{DL}}^2=
\max\!\left\{0,
\frac{Q_t-(n_t-1)}
{\sum_i a_i-\sum_i a_i^2/\sum_i a_i}\right\},
\quad
Q_t=\sum_i a_i(\hat\theta_i-\bar\theta_a)^2,
\label{eq:dl-heterogeneity}
\end{equation}

but small windows should pool heterogeneity across windows or use a hierarchical
model rather than treating this noisy estimate as known
\cite{dersimonianlaird1986}. Equation~\eqref{eq:index-variance} adds measurement
variance to an estimate of the deconvolved between-review variance; adding it to
an unadjusted residual variance would count measurement error twice.

The ordinary weight-concentration diagnostic is Kish's effective sample size
\cite{kish1965}:

\begin{equation}
n_{\mathrm{eff},t}=
\frac{\left(\sum_{i\in\mathcal{W}_t(\pi)}\omega_i\right)^2}
     {\sum_{i\in\mathcal{W}_t(\pi)}\omega_i^2}.
\label{eq:kish-effective-size}
\end{equation}

Dependence requires a second calculation. Define provisional clusters by
near-duplicate normalized text within the same language, version, and day; the
reference implementation uses five-token-shingle Jaccard similarity at 0.8 and
reports sensitivity at 0.7 and 0.9. Such screening is motivated by evidence that
review manipulation is structured rather than independent noise
\cite{mayzlin2014,lucazervas2016}. With residuals
$r_i=\hat\theta_i-I_t$ and $G_t>1$ clusters, use

\begin{equation}
\widehat V_{t,\mathrm{cl}}=
\frac{G_t}{G_t-1}
\frac{\sum_{g=1}^{G_t}\left(\sum_{i\in g}\omega_i r_i\right)^2}
     {\left(\sum_i\omega_i\right)^2},
\qquad
n_{\mathrm{eff},t}^{\mathrm{cl}}=
\frac{(\sum_i\omega_i)^2}
     {\sum_g(\sum_{i\in g}\omega_i)^2}.
\label{eq:cluster-variance}
\end{equation}

Report the design effect
$\mathrm{DEFF}_t=\widehat V_{t,\mathrm{cl}}/S_t^2$, the Fano factor of daily
arrivals, and concentration over version--locale--day cells. As operational
defaults, $\mathrm{DEFF}_t>2$ or a daily-count Fano factor above 2 triggers a
dependence flag, cluster-robust uncertainty, and suspension of shrinkage until a
campaign-level sensitivity analysis is reported. These thresholds are defaults
for auditability, not universal constants.

\subsubsection{Writing propensity and a comparability guardrail}

Repeated histogram snapshots permit the interval rating count $N_t$ and written
review count $n_t$ to be differenced, giving the writing propensity
$f_t=n_t/N_t$. Because submission timing is itself selected---early reviews come
from a different population than late ones, and the composition of the window
mean drifts with it \cite{lihitt2008}---$f_t$ is the observable that the
framework monitors rather than a nuisance parameter. If latent valence were
defined for all rating submitters, the
finite-population identity underlying the data-defect-correlation framework gives

\begin{equation}
\bar\theta_{W,t}-\bar\theta_{\mathrm{all},t}
=\rho_{R\theta,t}
 \sqrt{\frac{1-f_t}{f_t}}\,\sigma_{\theta,t},
\label{eq:data-defect-sensitivity}
\end{equation}

where the unobserved correlation $\rho_{R\theta,t}$ is varied over a declared
range rather than estimated from written reviews alone \cite{meng2018}. This is a
bias sensitivity band, not an identification result. With half-count smoothing
for zero cells, adjacent windows are flagged non-comparable when

\begin{equation}
\left|\operatorname{logit}(f_t)-\operatorname{logit}(f_{t-1})\right|
>2\sqrt{\frac{1}{n_t}+\frac{1}{N_t-n_t}
       +\frac{1}{n_{t-1}}+\frac{1}{N_{t-1}-n_{t-1}}}.
\label{eq:propensity-guardrail}
\end{equation}

No trend interpretation is made across a flagged boundary without a composition
analysis. This converts selection from a general caveat into a reproducible
abstention rule; it does not make the reviews representative. What remains
permitted is correspondingly narrow: a difference between two apps, two
countries, or two windows is a difference between two protocol-defined
collections, not by itself a difference between user populations. A level
comparison across units therefore additionally requires the same protocol $\pi$,
the same calibration and label set, and comparable writing propensity $f_t$, and
is reported only when those conditions are documented.

\subsection{Declared-prior shrinkage, empirical Bayes, and temporal smoothing}

Let $\Lambda_t=S_t^{-2}$ denote estimated data precision, let $C$ be a declared
global or category-level reference mean, and let $m$ be its prior precision. The
declared-prior conjugate estimate is

\begin{equation}
\widetilde I_t^{\mathrm{B}}
=\frac{\Lambda_t I_t+mC}{\Lambda_t+m}.
\label{eq:shrunk-sentiment-index}
\end{equation}

This has the algebraic form of the Gaussian posterior mean in
Section~\ref{sec:shrinkage}, with conditional variance
$(\Lambda_t+m)^{-1}$. It is not empirical Bayes when $(C,m)$ are declared, and it
is not an exact posterior when $\Lambda_t$ is estimated and treated as known.
Reports using this option must display the prior weight
$m/(\Lambda_t+m)$ and repeat the analysis at one quarter and four times the
declared ratio $m/\Lambda_t$.

Both shrinkage and filtering are defined for the writer-population target in
equation~\eqref{eq:review-estimand}, because both use $S_t^2$ and the
superpopulation component $\hat\kappa_t^2$. Applying them to the
finite-collection target in equation~\eqref{eq:finite-review-estimand} would
change the estimand: the prior would inject information about writers who are
not in the window into a quantity defined conditionally on the reviews in hand.
The two targets therefore carry different regularization rules, and the choice is
made once, before any window is inspected.

Low precision also has two distinct causes that the arithmetic treats
identically: a window may contain few reviews, or it may contain many reviews
from a genuinely polarized population, since $\hat\kappa_t^2$ enters
$S_t^2$ on equal terms with measurement variance. Both shrink more, but only the
first is a sample-size problem. Reports must print $n_t$ and
$n_{\mathrm{eff},t}$ beside $\Lambda_t$ so the reader can distinguish them. The
interval attached to equation~\eqref{eq:shrunk-sentiment-index} is the
conditional variance $(\Lambda_t+m)^{-1}$ of the shrunk estimate; it is not the
interval for $I_t$, which is $M_t^2$ or $S_t^2$ depending on the declared target,
and because the shrunk estimate is biased the conditional variance understates
its mean-squared error by the squared bias.

Genuine empirical Bayes requires a population of comparable apps. For app
$j=1,\ldots,J$ in a prespecified category, let

\begin{equation}
I_{jt}\mid\mu_{jt}\sim\mathcal N(\mu_{jt},S_{jt}^2),
\qquad
\mu_{jt}\sim\mathcal N(C_t,A_t).
\label{eq:cross-app-eb-model}
\end{equation}

Estimate $(C_t,A_t)$ jointly across apps by restricted maximum likelihood, which
is the declared default because it is defined for any $J\ge2$ and does not
inherit the small-unit undercoverage of moment estimators; a stated moment
estimator is reported as a sensitivity, because the two disagree when $J$ is
small, and the method actually used is named in the output. The estimate
$\widehat C_t$ is weighted by $(S_{jt}^2+\widehat A_t)^{-1}$. Conditional on the
estimated hyperparameters,

\begin{equation}
\widehat\mu_{jt}^{\mathrm{EB}}
=(1-B_{jt})I_{jt}+B_{jt}\widehat C_t,
\qquad
B_{jt}=\frac{S_{jt}^2}{S_{jt}^2+\widehat A_t}.
\label{eq:cross-app-eb-estimate}
\end{equation}

Intervals must account for hyperparameter estimation rather than using only the
conditional posterior variance \cite{efronmorris1975,morris1983}. This is also
the dimension in which the dominance result of Section~\ref{sec:shrinkage}
becomes operative rather than merely cautionary: with $J\ge3$ apps in a category
the James--Stein construction attains uniform risk dominance over the unpooled
app means \cite{stein1956,james1961}, which the one-dimensional single-app case
cannot. If comparable cross-app data are unavailable,
equation~\eqref{eq:shrunk-sentiment-index} remains
a declared-prior sensitivity analysis and is labelled accordingly.

Under the known-precision model of Section~\ref{sec:shrinkage}, shrinkage has
lower squared-error risk only when

\begin{equation}
|\mu_t-C|<\sqrt{\frac{1}{\Lambda_t}+\frac{2}{m}}.
\label{eq:practical-shrinkage-boundary}
\end{equation}

The unknown $\mu_t$ prevents this inequality from serving as a direct selection
rule. Large, precise departures from the reference mean---as may follow an
outage, review-bombing episode, controversial update, or successful relaunch---
are the cases most likely to fall outside the boundary. Estimated precision,
calibration bias, dependence, data-selected hyperparameters, and an obsolete
prior remove the exact risk guarantee. Section~\ref{sec:synthetic-study} therefore
checks the plug-in procedure by simulation instead of treating the boundary as
an observable decision rule.

Static pull toward $C$ and dynamic smoothing are reconciled by a single
mean-reverting state model. For elapsed time $\Delta_t$ between window midpoints,
use

\begin{equation}
I_t=\mu_t+\varepsilon_t,
\qquad \varepsilon_t\sim\mathcal{N}(0,S_t^2),
\qquad
\mu_t-C=\phi^{\Delta_t}(\mu_{t-1}-C)+\eta_t,
\qquad \eta_t\sim\mathcal{N}(0,Q_{\Delta_t}).
\label{eq:sentiment-state-model}
\end{equation}

For $|\phi|<1$, let the stationary variance be $A=1/m$ and set
$Q_{\Delta_t}=A(1-\phi^{2\Delta_t})$. The random-walk limit uses
$\phi=1$ and $Q_{\Delta_t}=Q\Delta_t$. Apply Section~\ref{sec:kalman} with
$R_t=S_t^2$ and with the innovation variance $\tau^2$ of that section replaced by
$Q_{\Delta_t}$, so that the notation of the two sections maps one to one. This
formulation uses the reference
distribution once rather than shrinking every window and then filtering the
already-shrunk series. The prospective published index is
$\widehat\mu_{t\mid t}$; a retrospective analysis may publish
$\widehat\mu_{t\mid T}$. The raw $I_t$ remains a required diagnostic.

Estimate $(\phi,Q)$ by maximizing the prediction-error likelihood

\begin{equation}
\ell(\phi,Q)=-\frac12\sum_t
\left\{\log F_t+v_t^2/F_t\right\}+\mathrm{const},
\quad
v_t=I_t-\widehat\mu_{t\mid t-1},
\quad F_t=P_{t\mid t-1}+S_t^2,
\label{eq:prediction-error-likelihood}
\end{equation}

with diffuse initialization when the initial state is not externally specified
\cite{durbinkoopman2012}. Profile likelihood or a weakly informative nonnegative
prior is used when the process variance piles up at zero. A parametric bootstrap
that simulates from the fitted series and re-estimates all state parameters is
used to assess the plug-in contribution to prediction mean-squared error
\cite{pfeffermanntiller2005}.

The Gaussian state is an unconstrained working scale and can leave $[-1,1]$. If
strictly bounded latent sentiment is required, use a link such as
$\mu_t=\tanh(z_t)$ and a nonlinear state-space method. Version boundaries can be
examined as interventions, but they remain associational unless competing trends
and changes in the reviewing population are addressed.

Gaussian adequacy concerns the window-level innovations, not the marginal shape
of individual star ratings. After fitting, monitor the standardized innovation

\begin{equation}
e_t^{\mathrm{std}}=
\frac{I_t-\hat\mu_{t\mid t-1}}
     {\sqrt{P_{t\mid t-1}+S_t^2}}.
\label{eq:standardized-innovation}
\end{equation}

Persistent skewness, heavy tails, multimodality, serial correlation, or runs of
extreme same-sign innovations indicate failure of the exact Gaussian MMSE
interpretation. A coordinated burst of nearly identical one-star reviews followed
by rapid recovery is a representative failure pattern. Prespecified alternatives
include a Student-$t$ observation model for isolated outliers, a change-point or
intervention model for release shocks, a mixture model for polarized populations,
and a nonlinear or particle filter for bounded non-Gaussian states. Under weaker
second-moment and uncorrelatedness conditions, the ordinary Kalman recursion may
retain a linear-projection interpretation, but it is no longer the exact posterior
mean.

\subsection{Distributional diagnostics and optional calibration}

The app-level histogram is cumulative. It is not commensurate with a single
review window unless category-specific snapshots are retained at both window
boundaries. Even with differenced counts, a nonzero gap mixes writing selection,
the acquisition and underreporting mechanisms that generate the J-shaped rating
distribution \cite{hupavlouzhang2017}, drift in quality over the window, and
growth of the rating base, so it is read as a description of composition rather
than as a test of sentiment. Let $c_{kt}^{\mathrm{all}}$ be the differenced
all-rating count and
$c_{kt}^{W}$ the complete written-review count in category $k$. Only when both
conditions hold may rating-only counts be reconstructed as
$c_{kt}^{RO}=c_{kt}^{\mathrm{all}}-c_{kt}^{W}$. Negative reconstructed cells flag
an incomplete or inconsistent collection and preclude the comparison.

Let $\widehat p_{kt}^{W}=n_{kt}/n_t$ and $\widehat p_{kt}^{RO}$ denote the
\emph{unweighted} category proportions of the written-review sample and of the
reconstructed rating-only counts. Every hat in this subsection is an unweighted
share; the weighted shares $\widetilde p_k$ defined below enter only the optional
calibration and never the gap statistic. The descriptive cumulative-distribution
gap is

\begin{equation}
D_{\star,t}=\max_{1\le k\le4}
\left|\sum_{j=1}^{k}\widehat p_{jt}^{W}-
      \sum_{j=1}^{k}\widehat p_{jt}^{RO}\right|.
\label{eq:discrete-cdf-gap}
\end{equation}

Although equation~\eqref{eq:discrete-cdf-gap} resembles a
Kolmogorov--Smirnov statistic, continuous-null critical values do not give the
exact discrete calibration. For a $K$-category multinomial sample,
Theorem~\ref{thm:multinomial-limit} gives

\begin{equation}
\sqrt n(\widehat{\bm p}-\bm p)
\rightsquigarrow
\mathcal N_K\!\left(\bm 0,
\operatorname{diag}(\bm p)-\bm p\bm p^\top\right),
\label{eq:multinomial-clt}
\end{equation}

and the covariance of cumulative coordinates is
$\min(F_j,F_k)-F_jF_k$. Critical values for
$D_{\star,t}$ are obtained by a multinomial parametric bootstrap under the pooled
null, or by a two-sample permutation test that shuffles the written versus
rating-only label conditional on the pooled star outcomes
\cite{conover1972,arnoldemerson2011}. This inference still requires an
exchangeable collection. For API-ranked subsets, report the gap descriptively
across sort orders, snapshots, and version overlays.

If optional base weights $\omega_i$ are already present, define their current
category share as

\begin{equation}
\widetilde p_k=
\frac{\sum_{i\in\mathcal{W}_t:r_i=k}\omega_i}
     {\sum_{i\in\mathcal{W}_t}\omega_i},
\end{equation}

the tilde marking a weighted share so that it cannot be confused with the
unweighted $\widehat p_k$ of equation~\eqref{eq:discrete-cdf-gap}; when no base
weights are present the two coincide.

To avoid infinite adjustments in sparse cells, define a smoothed written share

\begin{equation}
\widetilde p_{k,\alpha}=
\frac{n_{kt}+\alpha}{n_t+5\alpha},\qquad \alpha>0,
\label{eq:dirichlet-share}
\end{equation}

or collapse substantively adjacent cells under a prespecified rule
\cite{little1993}. With a declared target share $p_{kt}^{\mathrm{target}}$, the
calibrated weights are

\begin{equation}
\omega_i^{\mathrm{cal}}
=\omega_i\frac{p_{r_i,t}^{\mathrm{target}}}
                    {\widetilde p_{r_i,t,\alpha}}.
\label{eq:histogram-calibration-weight}
\end{equation}

This aligns weighted star-category shares but cannot recover rating-only text or
remove within-category selection. It is therefore a sensitivity analysis under
the untestable assumption that written and rating-only users have equal mean
valence within star category. Uncertainty must bootstrap both written records and
the interval histogram counts; equation~\eqref{eq:index-variance} may not treat
the estimated calibration weights as fixed.

\subsection{Implementation-neutral pseudocode}

Table~\ref{tab:index-pseudocode} summarizes one reproducible execution path,
and Figure~\ref{fig:pipeline} shows the same order. The
calibration functions and covariance model are fitted on independent validation
data; heterogeneity and state parameters follow the rules below. The accompanying
reference module, unit tests, collection script, synthetic harness, fixed random
seed, and generated CSV summaries implement the displayed arithmetic. They verify
reproducibility of the specification; the collected sample of
Section~\ref{sec:collected-sample} exercises it on real records, and neither
constitutes empirical validation on app users.

\begin{figure*}[htbp]
\centering
\includegraphics[width=0.88\textwidth]{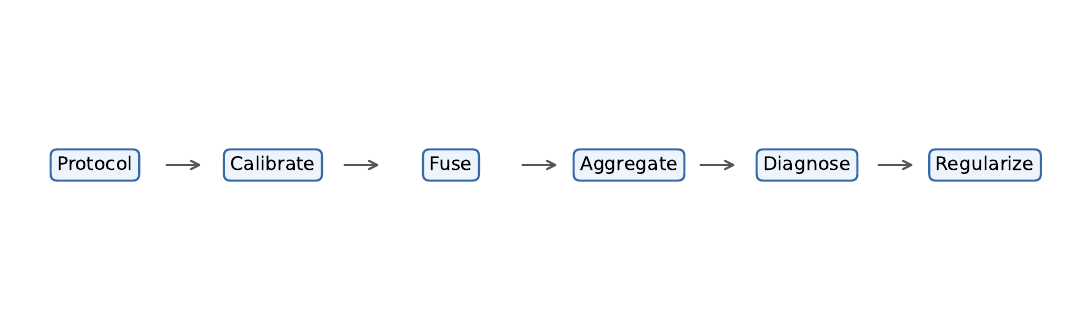}
\caption{Implemented framework stages. Protocol and calibration choices precede fusion;
target-specific aggregation and diagnostics precede any regularization.}
\label{fig:pipeline}
\end{figure*}

\begin{table*}[htbp]
\caption{Pseudocode for estimating and diagnosing the review-sentiment index.}
\label{tab:index-pseudocode}
\centering
\footnotesize
\begin{tabular}{>{\raggedright\arraybackslash}p{0.08\textwidth}
                >{\raggedright\arraybackslash}p{0.82\textwidth}}
\hline
\textbf{Stage} & \textbf{Operation} \\
\hline
Input & Supply review records $(r_i,s_i,u_i,t_i)$, disjoint windows
$\mathcal{W}_t(\pi)$, validation-fitted ordinal/text calibration and
$\bm\Sigma_i$, repeated rating-histogram snapshots, and the declared policy
parameters $(\alpha,q_{\max},\lambda)$. \\
1 & Record $\pi$, scraper version and commit, retrieved and app-level totals;
deduplicate \,\texttt{reviewId}, detect edits across snapshots, enforce timestamp
rules, and use fixed-duration calendar windows with version overlays. \\
2 & Fit or apply the ordinal--continuous measurement model. For the documented
linear approximation, compute calibrated $\bm y_i$, the constrained weights in
equation~\eqref{eq:constrained-fusion}, and
$(\hat\theta_i,\nu_i,\delta_i)$. Report non-identity calibration and held-out
coverage; otherwise label the fusion illustrative. \\
3 & Compute the primary unit-weighted $I_t$ and separately compute the policy
index from equation~\eqref{eq:review-weight}. Report the half-life and
$n_{\mathrm{eff},t}$ for each. \\
4 & Report $M_t^2$ for the finite records and $S_t^2$ for the writer population.
Pool $\kappa_t^2$ across small windows. Form duplicate-text clusters and replace
ordinary uncertainty by equation~\eqref{eq:cluster-variance} when
$\mathrm{DEFF}>2$ or the daily-count Fano factor exceeds 2. \\
5 & Compute $f_t$, the sensitivity band in
equation~\eqref{eq:data-defect-sensitivity}, and the guardrail in
equation~\eqref{eq:propensity-guardrail}. A flagged boundary is not interpreted
as a sentiment trend without composition analysis. \\
6 & Use cross-app empirical Bayes only when equation~\eqref{eq:cross-app-eb-model}
is estimable, meaning that the prespecified category supplies at least five apps
with a window-$t$ estimate and standard error. Otherwise report
equation~\eqref{eq:shrunk-sentiment-index} as a declared-prior sensitivity at
$m/\Lambda_t$ multiplied by $1/4$, 1, and 4, and never as the index. Skip
regularization in dependence-flagged windows. \\
7 & Estimate state parameters by equation~\eqref{eq:prediction-error-likelihood}
and filter the raw $I_t$. Publish $\widehat\mu_{t\mid t}$ prospectively or
$\widehat\mu_{t\mid T}$ retrospectively; retain $I_t$ as the raw diagnostic.
Bootstrap parameter uncertainty and diagnose standardized innovations. \\
8 & When interval snapshots and complete written counts are available, compute
$D_{\star,t}$ against the disjoint rating-only counts. Use a discrete bootstrap
or permutation calibration. Histogram reweighting is sensitivity-only and uses
equations~\eqref{eq:dirichlet-share}--\eqref{eq:histogram-calibration-weight}. \\
Output & Return one labelled primary series: the published index is
$\widehat\mu_{t\mid t}$ prospectively or $\widehat\mu_{t\mid T}$ retrospectively,
$I_t$ is never relabelled as the index, and shrunk values are returned only as
sensitivities. Attach to each series its own interval---$M_t^2$ or $S_t^2$ for the
raw index under the declared target, $(\Lambda_t+m)^{-1}$ for the declared-prior
estimate, the filter variance for the state series. Also return raw diagnostics,
star-only and text-only baselines, policy and prior sensitivities, selection and
dependence flags, nominal, Kish, and cluster effective sample sizes, and complete
provenance. \\
\hline
\end{tabular}
\end{table*}

\subsection{Toy calculation and interpretation}
\label{sec:toy-calculation}

Consider three explicitly synthetic reviews with mixed star categories. Use the
non-identity illustrative calibrations
$z_{\rho i}=-0.02+0.95\rho_i$ and $z_{s i}=0.03+0.90s_i$ and the declared
covariance $\bm\Sigma=\operatorname{diag}(0.03,0.05)$. Equation
\eqref{eq:constrained-fusion} then gives
$a_{\rho}=0.05/(0.03+0.05)=0.625$, $a_s=0.375$, and
$\nu=0.01875$. These quantities are chosen to expose every arithmetic step; they
are not fitted classifier outputs. Table~\ref{tab:worked-sentiment} displays the
review-level quantities that follow.

\begin{table*}[htbp]
\caption{Synthetic three-review calculation using calibrated signals.}
\label{tab:worked-sentiment}
\centering
\footnotesize
\begin{tabular}{>{\raggedright\arraybackslash}p{0.13\textwidth}
                r r r r r r}
\hline
\textbf{Review} & \textbf{Stars} & $z_{\rho i}$ & $z_{s i}$ &
$\hat\theta_i$ & $\delta_i$ & $q_i$ \\
\hline
Synthetic 1 & 5 & 0.930 & 0.300 & 0.694 & 0.630 & 6.481 \\
Synthetic 2 & 2 & $-0.495$ & $-0.420$ & $-0.467$ & 0.075 & 3.565 \\
Synthetic 3 & 4 & 0.455 & 0.840 & 0.599 & 0.385 & 5.762 \\
\hline
\end{tabular}
\end{table*}

The primary unit-weighted fused index is

\begin{equation}
I_{\mathrm{unweighted}}
=\frac{0.69375-0.466875+0.599375}{3}=0.2754.
\label{eq:toy-unweighted}
\end{equation}

The normalized-star and calibrated-text baselines are 0.3333 and 0.2400,
respectively. With helpfulness counts $(239,12,116)$, $\alpha=1$,
$q_{\max}=6$, and $\lambda=0$, equation~\eqref{eq:review-weight} gives the
displayed $q_i=1+\log(1+u_i)$ and

\begin{equation}
I_{\mathrm{policy}}=
\frac{6.4806(0.69375)+3.5649(-0.466875)+5.7622(0.599375)}
     {6.4806+3.5649+5.7622}=0.3976,
\label{eq:toy-policy}
\end{equation}

with $n_{\mathrm{eff}}=2.84$. The policy choice raises the result by 0.1222
because the more positive synthetic records receive more prominence. That gap is
the substantive effect of changing the target and is why the policy result may
not replace the primary unit-weighted result. The largest disagreement belongs to
Synthetic 1 despite its five-star rating. A high-star review containing both
genuine praise and a severe version-specific complaint remains unresolved by any
scalar fusion \cite{pang2008opinion,liu2012sentiment}: the method can flag the
conflict but needs aspect- and
version-aware labels to interpret it.

\subsection{Reproducible synthetic stress study}
\label{sec:synthetic-study}

The included harness uses seed 20261003 and 1,200 replications per static
scenario. Review volume follows a negative-binomial distribution around the
scenario mean. Latent writer valence is Gaussian or heavy-tailed; review
submission may depend logistically on valence; star ratings threshold a noisy
latent score; the text signal is a saturating hyperbolic-tangent score with
valence-dependent noise; and one scenario overwrites one quarter of the signals
with coordinated negative records. The dynamic experiment uses 500 replications
of a random walk with a level step and an isolated outage. These simulations test
the implementation against known synthetic targets; they do not validate the
measurement model for Google Play. Baselines include star-only, text-only,
unweighted fusion, a naive count-weighted Bayesian average with eight prior
pseudo-reviews, and precision-weighted declared-prior shrinkage.

\begin{table*}[htbp]
\caption{Selected deterministic simulation results. Coverage is the nominal
95\% interval coverage for the realized protocol-defined writer mean.}
\label{tab:simulation-summary}
\centering
\footnotesize
\begin{tabular}{l l r r r r}
\hline
\textbf{Scenario} & \textbf{Method} & \textbf{Bias} & \textbf{RMSE} &
\textbf{Coverage} & \textbf{Mean width} \\
\hline
Clean & Star only & $-0.004$ & 0.040 & 0.998 & 0.265 \\
Clean & Text only & $-0.036$ & 0.049 & 0.985 & 0.221 \\
Clean & Fusion & $-0.016$ & 0.032 & 1.000 & 0.232 \\
Clean & Count Bayes & $-0.059$ & 0.066 & 0.896 & 0.187 \\
Clean & Declared prior & $-0.023$ & 0.035 & 1.000 & 0.228 \\
Small $E(n)=8$ & Fusion & $-0.014$ & 0.063 & 0.998 & 0.511 \\
Small $E(n)=8$ & Count Bayes & $-0.137$ & 0.168 & 0.439 & 0.231 \\
Small $E(n)=8$ & Declared prior & $-0.042$ & 0.078 & 0.997 & 0.467 \\
Coordinated attack & Fusion & $-0.319$ & 0.321 & 0.024 & 0.402 \\
Coordinated attack & Declared prior & $-0.313$ & 0.315 & 0.019 & 0.385 \\
\hline
\end{tabular}
\end{table*}

Table~\ref{tab:simulation-summary} reports selected rows from these runs. The
clean scenario shows a modest RMSE advantage for fusion. The naive
count-weighted baseline is overconfident, and the declared prior increases RMSE
in the small-window setting because its reference mean is wrong. More
importantly, the coordinated block defeats every ordinary interval. High volume
does not make dependent manipulation reliable. This is the
reason for the duplicate-cluster abstention rule rather than a claim that
shrinkage repairs a campaign.

\begin{figure*}[htbp]
\centering
\begin{minipage}{0.48\textwidth}
\centering
\includegraphics[width=\textwidth]{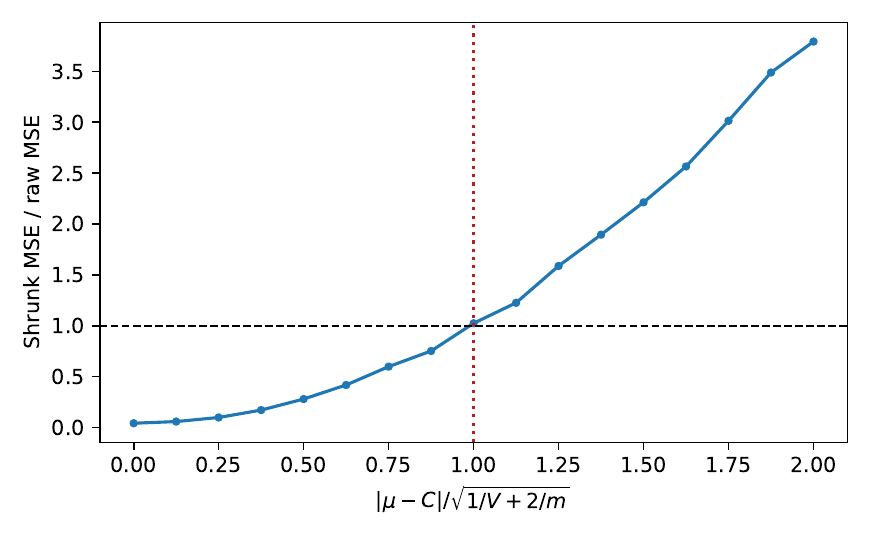}
\end{minipage}\hfill
\begin{minipage}{0.48\textwidth}
\centering
\includegraphics[width=\textwidth]{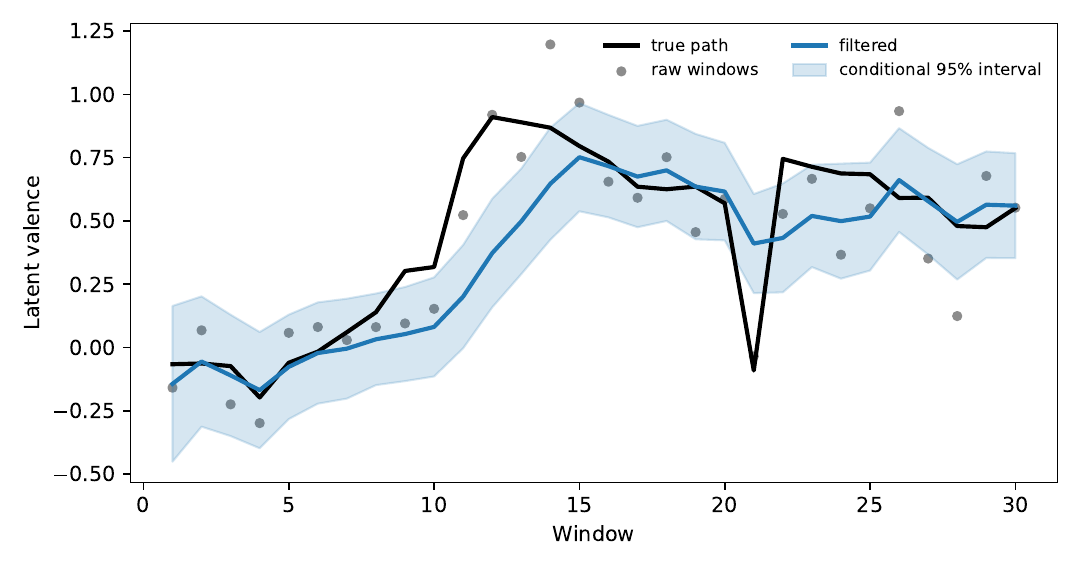}
\end{minipage}
\caption{Left: simulated shrinkage-to-raw MSE ratio against the theoretical
risk-radius scale. Right: one synthetic filtered path with conditional intervals.}
\label{fig:simulation-diagnostics}
\end{figure*}

Figure~\ref{fig:simulation-diagnostics} shows the risk boundary against the
theoretical radius on the left and one filtered path with its conditional
interval on the right. The empirical MSE crossover occurs near the theoretical
unit boundary: the simulated ratio is 1.02 at
$|\mu-C|/\sqrt{1/\Lambda+2/m}=1$. In the dynamic study, filtering reduces RMSE
from 0.196 to 0.166, but its nominal interval covers only 0.859 of the latent
states. The latter result demonstrates why state-parameter and structural
uncertainty must be bootstrapped rather than hidden by the exact conditional
Kalman algebra. Complete rows, including selection, heavy-tail, sensitivity, and
held-out disagreement-inflation grids, are distributed as CSV files and
regenerated by the documented command.

\subsection{A collected-sample application}
\label{sec:collected-sample}

The framework was also run on a real collection made under the declared protocol
$\pi=(\ell,c,o,\mathcal{W},L)=(\mathrm{en},\mathrm{us},\text{newest},
\text{calendar day},2000)$ with the vendored \texttt{google-play-scraper}
release 1.2.7 (commit \texttt{ce1df6d}). Three apps were requested---Google
Maps, Spotify, and
WhatsApp---and each returned exactly $L=2000$ written reviews over ten pages on
2026-10-04 UTC, so every app is truncated at the cap and every sample is a
selected collection. The 6000 retained records carry 6000 distinct
\texttt{reviewId} values, so de-duplication dropped none of them. The
contemporaneous app-level totals recorded at collection time run from
$19{,}581{,}159$ ratings and $678{,}106$ reviews for Google Maps to
$244{,}995{,}204$ ratings and $1{,}952{,}593$ reviews for WhatsApp, which fixes
the descriptive written share at 3.5\%, 5.1\%, and 0.8\% of rating submitters.
Those totals are not the window-level writing propensity $f_t$: one cumulative
snapshot cannot be differenced, so the guardrail in
equation~\eqref{eq:propensity-guardrail} is reported below as unavailable rather
than as passed.

The measurement stage uses the identity calibrations $z_{\rho i}=(r_i-3)/2$ and
the VADER compound score \cite{hutto2014vader} on the same $[-1,1]$ scale, with
the declared covariance $\bm\Sigma=\operatorname{diag}(0.03,0.05)$ of
Section~\ref{sec:toy-calculation} carried over unchanged. The fusion weights are
therefore $a_\rho=0.625$, $a_s=0.375$, and $\nu=0.01875$: the worked example
applied to real signals. No independently labelled validation records exist, so
those quantities remain declarations. The fusion is labelled illustrative, the
covariance is not fitted, and the held-out coverage required by Stage~2 of
Table~\ref{tab:index-pseudocode} is recorded as a non-result instead of as a
passed check.

Calendar-day windows give 13 windows across the three apps. The first and last
day of each sample are partial by construction---the oldest day is cut by the
retrieval cap and the newest by the collection time---and those windows are
flagged rather than silently averaged over.
Table~\ref{tab:collected-summary} summarizes the resulting series.

\begin{table*}[htbp]
\caption{Collected-sample summary by app. $\bar I_t$ is the mean raw window
index, Star only and Text only are the mandatory baselines, App score is the
contemporaneous app-detail score on the one-to-five scale, $\hat\kappa_t^2$ is
pooled across the app's windows by the DerSimonian--Laird diagnostic, DEFF is
the maximum near-duplicate design effect over windows, and Fano is the arrival
count factor over interior windows only.}
\label{tab:collected-summary}
\centering
\footnotesize
\begin{tabular}{l r r r r r r r r r}
\hline
& & \textbf{Windows} & & \textbf{Star} & \textbf{Text} & \textbf{App} &
& \textbf{max} & \\
\textbf{App} & $n$ & (interior) & $\bar I_t$ & only & only & score &
$\hat\kappa_t^2$ & DEFF & Fano \\
\hline
Google Maps & 2000 & 7 (5) & 0.310 & 0.366 & 0.216 & 3.217 & 0.407 & 1.33 & 11.2 \\
Spotify & 2000 & 4 (2) & 0.399 & 0.450 & 0.314 & 4.347 & 0.347 & 1.15 & -- \\
WhatsApp & 2000 & 2 (0) & 0.476 & 0.589 & 0.287 & 4.617 & 0.247 & 1.61 & -- \\
\hline
\end{tabular}
\end{table*}

\begin{figure*}[htbp]
\centering
\includegraphics[width=0.95\textwidth]{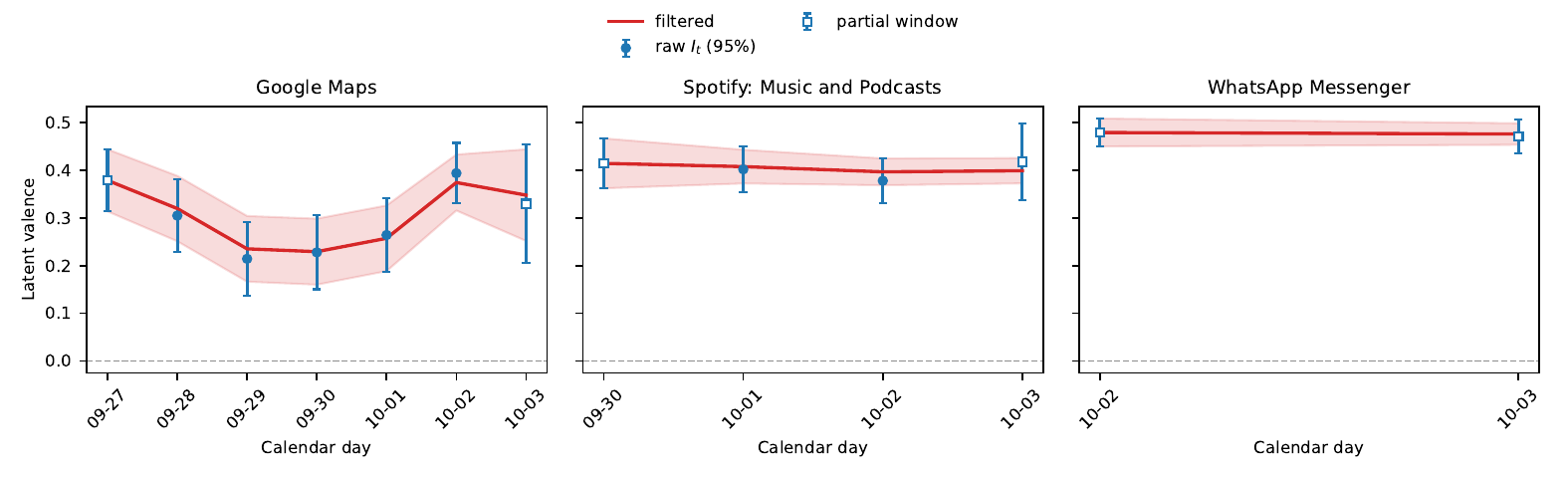}
\caption{Collected-sample window index with 95\% writer-population intervals
(blue) and the filtered series with conditional intervals (red). Open squares
mark the first and last day of each sample, which are partial by construction;
the dashed line is neutral valence.}
\label{fig:collected-index}
\end{figure*}

Figure~\ref{fig:collected-index} shows the raw series and the filtered series.
The star--text gap is visible in every app: text-only scores sit $0.14$ to $0.30$
below star-only scores, so mean disagreement $\bar\delta$ is $0.60$, $0.55$, and
$0.66$, and between 46\% and 65\% of reviews exceed the declared threshold
$c_\delta=0.5$. Applying equation~\eqref{eq:disagreement-inflation} with the
harness-selected pair $(c_\delta,\gamma)=(0.5,12)$ widens $S_t$ by 1.6--6.0\%,
consistent with the matched arm of that study: it is reported as a sensitivity,
not as a correction, and it does not decide which signal is wrong.

Uncertainty is dominated by between-review variation rather than measurement
error. For example, the Google Maps window of 2026-09-29 has
$M_t=0.008$ against $S_t=0.039$, because the pooled
$\hat\kappa_t^2=0.407$ enters $S_t^2$ on equal terms with $\nu=0.01875$ and is
then divided by $n_t=276$. Unit weights make the nominal, Kish, and cluster
effective sample sizes differ only through clustering: the near-duplicate design
effect never exceeds 1.61 across the 13 windows, below the dependence default of
2, while the interior arrival-count Fano factor for Google Maps is 11.2 and
triggers the arrival-dependence flag. Spotify and WhatsApp have too few interior
windows for that factor to be defined, which is recorded as a non-result.

Regularization is applied only where the stages of
Table~\ref{tab:index-pseudocode} permit it. Cross-app empirical Bayes is not
estimable: Stage~6 requires at least five apps in the prespecified category with
a window estimate and standard error, and the panel has three. The
declared-prior sensitivity with $C=0$ is therefore reported at
$m/\Lambda_t\in\{1/4,1,4\}$, that is, at prior weights $0.20$, $0.50$, and
$0.80$; the Google Maps window of 2026-10-02 moves from $I_t=0.394$ to $0.315$,
$0.197$, and $0.079$ across that grid, which is the size of the reporting choice
being made. State parameters are estimated by the prediction-error likelihood of
equation~\eqref{eq:prediction-error-likelihood}: Google Maps gives
$\hat Q=5.0\times10^{-3}$ with profile interval
$[1.6\times10^{-3},1.3\times10^{-2}]$ and standardized innovations within
$\pm1.6$, while Spotify and WhatsApp pile up at the lower grid boundary
$\hat Q=10^{-6}$, a boundary estimate reported as such because two to four
windows identify the process variance barely if at all. Filter intervals are
plug-in conditional intervals; the parametric bootstrap of Stage~7 is not run,
because refitting state parameters inside a bootstrap is not supported by so few
windows.

All quantities unavailable under the protocol are reported:
\texttt{non-results.csv} records the writing propensity and its guardrail, the
discrete written-versus-rating gap $D_{\star,t}$ of
equation~\eqref{eq:discrete-cdf-gap}, cross-app empirical Bayes, held-out
calibration coverage, the bootstrap correction, and the interior Fano factor for
the two short samples, each with the reason it is unavailable. The result is a
descriptive application of the framework under a declared $\pi$, not a validated
index: nothing here establishes that the identity calibrations are unbiased, and
the three apps differ in app-level score, written share, and sample length, so a
level comparison across them would additionally require the same calibration and
label set before it could be read as a difference in writer populations.

\subsection{Minimum reporting and validation requirements}

An empirical application must state whether its target is the finite collection
in equation~\eqref{eq:finite-review-estimand} or the writer population in
equation~\eqref{eq:review-estimand}, and pair it with $M_t^2$ or $S_t^2$ rather
than mixing the two. It must report the primary unit-weighted estimate, every
policy-weighted alternative, the published filtered or smoothed series, and the
raw diagnostic series; nominal, Kish, and cluster effective sample sizes;
$\hat\kappa_t^2$, DEFF, burstiness, writing propensity, comparability flags, and
the data-defect sensitivity band; and $D_{\star,t}$ only when interval-specific
rating-only counts can be reconstructed.

Reproducibility metadata comprise language, locale, country, sort order, exact
query times, fixed window duration, version overlays, retrieval cap, number
exposed versus app totals, scraper version and commit, edit detection,
missingness, timestamp, and deduplication rules. Validation reports cover
label construction, ordinal thresholds, text calibration, covariance and
heteroskedasticity estimation, language/domain transfer, and held-out coverage.
Sensitivity settings include $(\alpha,q_{\max},\lambda)$,
$(c_{\delta},\gamma)$, duplicate thresholds, cell-smoothing $\alpha$, the
correlation range in equation~\eqref{eq:data-defect-sensitivity}, prior weight,
heterogeneity pooling, and state parameters.

\noindent\textbf{Recommended defaults.} The parameters with the largest leverage
are the prior weight $m/\Lambda_t$, the process variance $Q$, and the policy
triple $(\alpha,q_{\max},\lambda)$. The defaults are: no static shrinkage in the
published series unless cross-app empirical Bayes is estimable under Stage~6 of
Table~\ref{tab:index-pseudocode}, because otherwise the prior is declared rather
than estimated; $Q$ estimated by equation~\eqref{eq:prediction-error-likelihood}
with the bootstrap correction of Stage~7 and never set by hand, because a fixed
$Q$ presumes equally spaced windows; $\lambda=0$ unless a half-life no longer
than half the window duration is substantively required, which keeps the
truncated recency kernel from concentrating on the window edge; and $\alpha=1$
with $q_{\max}=6$ only as a display convention for the policy series, which is
always reported beside the unit-weighted index. The mandatory sensitivity grid
crosses $m/\Lambda_t\in\{1/4,1,4\}$, $Q$ at the endpoints of its
profile-likelihood interval, at least two window rules, and the duplicate-text
thresholds at $0.7$, $0.8$, and $0.9$; \texttt{sensitivity-grid.csv} supplies the
template. Under these defaults, two analysts working from the same protocol
$\pi$, the same validation fit, and the same window rule report the same index up
to sampling error, which is the constraint that declaring free parameters without
defaults does not provide.

Star-only, text-only, and
unweighted-fusion baselines are mandatory. Any threshold changed after examining
the outcome is labelled exploratory.

%% file: sn-bibliography.bib
@book{billingsley1999,
  author = {Billingsley, Patrick},
  title = {Convergence of Probability Measures},
  edition = {2},
  publisher = {Wiley},
  address = {New York},
  year = {1999}
}

@book{billingsley1968,
  author = {Billingsley, Patrick},
  title = {Convergence of Probability Measures},
  publisher = {Wiley},
  address = {New York},
  year = {1968}
}

@book{shorackwellner1986,
  author = {Shorack, Galen R. and Wellner, Jon A.},
  title = {Empirical Processes with Applications to Statistics},
  publisher = {Wiley},
  address = {New York},
  year = {1986}
}

@book{dudley2014,
  author = {Dudley, R. M.},
  title = {Uniform Central Limit Theorems},
  edition = {2},
  publisher = {Cambridge University Press},
  address = {Cambridge},
  year = {2014}
}

@book{vandervaart1998,
  author = {van der Vaart, A. W.},
  title = {Asymptotic Statistics},
  publisher = {Cambridge University Press},
  address = {Cambridge},
  year = {1998}
}

@article{aitken1936,
  author = {Aitken, A. C.},
  title = {On least squares and linear combination of observations},
  journal = {Proceedings of the Royal Society of Edinburgh},
  volume = {55},
  pages = {42--48},
  year = {1936}
}

@article{stein1956,
  author = {Stein, Charles},
  title = {Inadmissibility of the usual estimator for the mean of a multivariate normal distribution},
  journal = {Proceedings of the Third Berkeley Symposium on Mathematical Statistics and Probability},
  volume = {1},
  pages = {197--206},
  year = {1956}
}

@article{james1961,
  author = {James, W. and Stein, Charles},
  title = {Estimation with quadratic loss},
  journal = {Proceedings of the Fourth Berkeley Symposium on Mathematical Statistics and Probability},
  volume = {1},
  pages = {361--379},
  year = {1961}
}

@article{kalman1960,
  author = {Kalman, R. E.},
  title = {A new approach to linear filtering and prediction problems},
  journal = {Journal of Basic Engineering},
  volume = {82},
  number = {1},
  pages = {35--45},
  year = {1960}
}

@article{pang2008opinion,
  author = {Pang, Bo and Lee, Lillian},
  title = {Opinion mining and sentiment analysis},
  journal = {Foundations and Trends in Information Retrieval},
  volume = {2},
  number = {1--2},
  pages = {1--135},
  year = {2008},
  doi = {10.1561/1500000011}
}

@book{liu2012sentiment,
  author = {Liu, Bing},
  title = {Sentiment Analysis and Opinion Mining},
  series = {Synthesis Lectures on Human Language Technologies},
  volume = {5},
  number = {1},
  pages = {1--167},
  publisher = {Morgan \& Claypool},
  address = {San Rafael, CA},
  year = {2012},
  doi = {10.2200/S00416ED1V01Y201204HLT016}
}

@article{anscombe1948,
  author = {Anscombe, F. J.},
  title = {The transformation of {P}oisson, binomial and negative-binomial data},
  journal = {Biometrika},
  volume = {35},
  number = {3--4},
  pages = {246--254},
  year = {1948},
  doi = {10.1093/biomet/35.3-4.246}
}

@inproceedings{hutto2014vader,
  author = {Hutto, C. J. and Gilbert, Eric},
  title = {{VADER}: A parsimonious rule-based model for sentiment analysis of social media text},
  booktitle = {Proceedings of the Eighth International AAAI Conference on Weblogs and Social Media},
  volume = {8},
  number = {1},
  pages = {216--225},
  year = {2014},
  doi = {10.1609/icwsm.v8i1.14550}
}

@article{wilson1927,
  author = {Wilson, Edwin B.},
  title = {Probable inference, the law of succession, and statistical inference},
  journal = {Journal of the American Statistical Association},
  volume = {22},
  number = {158},
  pages = {209--212},
  year = {1927},
  doi = {10.1080/01621459.1927.10502953}
}

@misc{googlePlayScraper,
  author = {{JoMingyu}},
  title = {google-play-scraper: Google Play scraper for Python},
  howpublished = {GitHub repository and package documentation},
  url = {https://github.com/JoMingyu/google-play-scraper},
  note = {Version and commit are recorded with each empirical collection; no collection was performed for this manuscript; accessed 26 September 2026}
}

@article{mccullagh1980,
  author = {McCullagh, Peter},
  title = {Regression models for ordinal data},
  journal = {Journal of the Royal Statistical Society: Series B},
  volume = {42},
  number = {2},
  pages = {109--127},
  year = {1980},
  doi = {10.1111/j.2517-6161.1980.tb01109.x}
}

@article{meng2018,
  author = {Meng, Xiao-Li},
  title = {Statistical paradises and paradoxes in big data ({I}): Law of large populations, big data paradox, and the 2016 {US} presidential election},
  journal = {The Annals of Applied Statistics},
  volume = {12},
  number = {2},
  pages = {685--726},
  year = {2018},
  doi = {10.1214/18-AOAS1161SF}
}

@article{conover1972,
  author = {Conover, William J.},
  title = {A Kolmogorov goodness-of-fit test for discontinuous distributions},
  journal = {Journal of the American Statistical Association},
  volume = {67},
  number = {339},
  pages = {591--596},
  year = {1972},
  doi = {10.1080/01621459.1972.10481254}
}

@article{arnoldemerson2011,
  author = {Arnold, Taylor B. and Emerson, John W.},
  title = {Nonparametric goodness-of-fit tests for discrete null distributions},
  journal = {The R Journal},
  volume = {3},
  number = {2},
  pages = {34--39},
  year = {2011},
  doi = {10.32614/RJ-2011-016}
}

@article{efronmorris1975,
  author = {Efron, Bradley and Morris, Carl},
  title = {Data analysis using Stein's estimator and its generalizations},
  journal = {Journal of the American Statistical Association},
  volume = {70},
  number = {350},
  pages = {311--319},
  year = {1975}
}

@article{morris1983,
  author = {Morris, Carl N.},
  title = {Parametric empirical Bayes inference: Theory and applications},
  journal = {Journal of the American Statistical Association},
  volume = {78},
  number = {381},
  pages = {47--55},
  year = {1983}
}

@article{dersimonianlaird1986,
  author = {DerSimonian, Rebecca and Laird, Nan},
  title = {Meta-analysis in clinical trials},
  journal = {Controlled Clinical Trials},
  volume = {7},
  number = {3},
  pages = {177--188},
  year = {1986}
}

@book{durbinkoopman2012,
  author = {Durbin, James and Koopman, Siem Jan},
  title = {Time Series Analysis by State Space Methods},
  edition = {2},
  publisher = {Oxford University Press},
  address = {Oxford},
  year = {2012}
}

@article{pfeffermanntiller2005,
  author = {Pfeffermann, Danny and Tiller, Richard},
  title = {Bootstrap approximation to prediction {MSE} for state-space models with estimated parameters},
  journal = {Journal of Time Series Analysis},
  volume = {26},
  number = {6},
  pages = {893--916},
  year = {2005},
  doi = {10.1111/j.1467-9892.2005.00448.x}
}

@book{kish1965,
  author = {Kish, Leslie},
  title = {Survey Sampling},
  publisher = {Wiley},
  address = {New York},
  year = {1965}
}

@article{little1993,
  author = {Little, Roderick J. A.},
  title = {Post-stratification: A modeler's perspective},
  journal = {Journal of the American Statistical Association},
  volume = {88},
  number = {423},
  pages = {1001--1012},
  year = {1993}
}

@inproceedings{paganomaalej2013,
  author = {Pagano, Dennis and Maalej, Walid},
  title = {User feedback in the appstore: An empirical study},
  booktitle = {2013 IEEE 21st International Requirements Engineering Conference},
  pages = {125--134},
  year = {2013}
}

@inproceedings{guzmanmaalej2014,
  author = {Guzman, Emitza and Maalej, Walid},
  title = {How do users like this feature? A fine-grained sentiment analysis of app reviews},
  booktitle = {2014 IEEE 22nd International Requirements Engineering Conference},
  pages = {153--162},
  year = {2014}
}

@inproceedings{maalejnabil2015,
  author = {Maalej, Walid and Nabil, Hadeer},
  title = {Bug report, feature request, or simply praise? On automatically classifying app reviews},
  booktitle = {2015 IEEE 23rd International Requirements Engineering Conference},
  pages = {116--125},
  year = {2015}
}

@article{martin2017,
  author = {Martin, William and Sarro, Federica and Jia, Yue and Zhang, Yuanyuan and Harman, Mark},
  title = {A survey of app store analysis for software engineering},
  journal = {IEEE Transactions on Software Engineering},
  volume = {43},
  number = {9},
  pages = {817--847},
  year = {2017}
}

@article{hupavlouzhang2017,
  author = {Hu, Nan and Pavlou, Paul A. and Zhang, Jennifer},
  title = {On self-selection biases in online product reviews},
  journal = {MIS Quarterly},
  volume = {41},
  number = {2},
  pages = {449--471},
  year = {2017}
}

@article{lihitt2008,
  author = {Li, Xinxin and Hitt, Lorin M.},
  title = {Self-selection and information role of online product reviews},
  journal = {Information Systems Research},
  volume = {19},
  number = {4},
  pages = {456--474},
  year = {2008}
}

@article{lucazervas2016,
  author = {Luca, Michael and Zervas, Georgios},
  title = {Fake it till you make it: Reputation, competition, and Yelp review fraud},
  journal = {Management Science},
  volume = {62},
  number = {12},
  pages = {3412--3427},
  year = {2016}
}

@article{mayzlin2014,
  author = {Mayzlin, Dina and Dover, Yaniv and Chevalier, Judith},
  title = {Promotional reviews: An empirical investigation of online review manipulation},
  journal = {American Economic Review},
  volume = {104},
  number = {8},
  pages = {2421--2455},
  year = {2014},
  doi = {10.1257/aer.104.8.2421}
}

@article{dellarocas2003,
  author = {Dellarocas, Chrysanthos},
  title = {The digitization of word of mouth: Promise and challenges of online feedback mechanisms},
  journal = {Management Science},
  volume = {49},
  number = {10},
  pages = {1407--1424},
  year = {2003},
  doi = {10.1287/mnsc.49.10.1407.17308}
}

@article{dabrowskiletier2022,
  author = {D{\k{a}}browski, Jacek and Letier, Emmanuel and Perini, Anna and Susi, Angelo},
  title = {Analysing app reviews for software engineering: A systematic literature review},
  journal = {Empirical Software Engineering},
  volume = {27},
  number = {2},
  pages = {43},
  year = {2022},
  doi = {10.1007/s10664-021-10065-7}
}

@article{nakamura2022,
  author = {Nakamura, Walter T. and de Oliveira, Edson Cesar and de Oliveira, Elaine H. T. and Redmiles, David and Conte, Tayana},
  title = {What factors affect the {UX} in mobile apps? A systematic mapping study on the analysis of app store reviews},
  journal = {Journal of Systems and Software},
  volume = {193},
  pages = {111462},
  year = {2022},
  doi = {10.1016/j.jss.2022.111462}
}

@inproceedings{devlin2019bert,
  author = {Devlin, Jacob and Chang, Ming-Wei and Lee, Kenton and Toutanova, Kristina},
  title = {{BERT}: Pre-training of deep bidirectional transformers for language understanding},
  booktitle = {Proceedings of the 2019 Conference of the North American Chapter of the Association for Computational Linguistics: Human Language Technologies, Volume 1},
  pages = {4171--4186},
  year = {2019},
  doi = {10.18653/v1/N19-1423}
}

@inproceedings{reimers2019sbert,
  author = {Reimers, Nils and Gurevych, Iryna},
  title = {Sentence-{BERT}: Sentence embeddings using {S}iamese {BERT}-networks},
  booktitle = {Proceedings of the 2019 Conference on Empirical Methods in Natural Language Processing and the 9th International Joint Conference on Natural Language Processing},
  pages = {3982--3992},
  year = {2019},
  doi = {10.18653/v1/D19-1410}
}
